\documentclass[mathematics,article,accept,pdftex,oneauthor]{mdpi_template/Definitions/mdpi}

\firstpage{1}
\pubvolume{1}
\issuenum{1}
\articlenumber{0}
\pubyear{2026}
\copyrightyear{2026}
\datereceived{ }
\daterevised{ }
\dateaccepted{ }
\datepublished{ }

\makeatletter
\let\@origmaketitle\@maketitle
\renewcommand{\@maketitle}{%
  \let\origincludegraphics\includegraphics
  \renewcommand{\includegraphics}[2][]{%
    \def\@tempa{##2}%
    \@expandtwoargs\in@{.eps}{\@tempa}%
    \ifin@ \else \origincludegraphics[##1]{##2}\fi
  }%
  \@origmaketitle
  \let\includegraphics\origincludegraphics
}
\fancypagestyle{plain}{%
  \fancyhf{}%
  \fancyfoot[C]{\footnotesize\thepage}%
}
\renewcommand{\cright}{}
\renewcommand{\leftcolDates}{}
\renewcommand{\leftcolumnUpdatelogo}{}
\renewcommand{\leftcolumnCorrstatement}{}
\renewcommand{\leftcolumnCright}{}
\makeatother

\newgeometry{left=2.5cm, right=2.5cm, top=1.27cm, bottom=1.08cm,
  includehead, includefoot, footskip=36pt, headsep=24pt}
\fancyheadoffset{0pt}
\fancyfootoffset{0pt}
\usepackage{siunitx}
\usepackage{placeins}
\DeclareMathOperator{\erfcx}{erfcx}
\DeclareMathOperator{\erfc}{erfc}
\DeclareMathOperator{\erf}{erf}

\Title{Faster Monotone Implied Volatility Solver}

\Author{Fabien Le Floc'h\orcidA{}}

\AuthorNames{Fabien Le Floc'h}

\address{%
Independent researcher; fabien@2ipi.com}

\corres{Correspondence: fabien@2ipi.com}

\abstract{%
We present ThiopheneIV, a Black--Scholes implied-volatility solver with a monotone core and explicit production guards. The solver starts from the simple  Choi--Huh--Su L3 lower-bound seed and applies three 
Euler--Chebyshev steps on a lower branch and three Halley steps on the remaining upper branch.  We prove
that, in exact arithmetic, the seed lies below the root and both maps increase
monotonically without overshooting.  We also detail the practical challenges encountered for a double-precision implementation: parity normalisation, microscopic Bachelier-limit handling, saturated price treatment, and an optional
Jäckel–Newton polish.  Across standard grids,
market-like data, high-volatility cases, and
adversarial corners, ThiopheneIV agrees closely with multiprecision Black
reference prices at low latency.  We provide detailed comparisons with recent
solvers, including J\"ackel's \emph{Let's Be Rational}.  The broader lesson is
that a convergence proof gives a clean core, but robust production inversion
still depends on boundary handling and on the pricing objective one chooses to
match.}


\begin{document}

\section{Introduction}
\label{sec:intro}

Implied volatility inversion recovers the Black--Scholes volatility $\sigma$
from an option price.  It is a small scalar root-finding problem, but it is
called so often in calibration and risk systems that both latency and tail
robustness matter.  The numerical difficulty is not merely local convergence of
Newton-like iterations.  Deep out-of-the-money prices can underflow, nearly
in-the-money prices can lose time value before the solver is called, and
high-volatility prices close to the upper bound are poorly conditioned in
volatility.

J\"ackel's \emph{Let's Be Rational}~\cite{jaeckel2017} remains the natural
reference implementation for this domain.  It combines normalised prices,
region-dependent asymptotics, complementary objectives, and carefully selected
iterations.  Choi, Huh, and Su~\cite{choi2023} derive implied-volatility bounds
from option-delta
inequalities and prove monotone convergence of Newton's method from their lower
bounds.

This paper studies a production solver built around the Choi--Huh--Su L3 lower
bound.  We call it ThiopheneIV.  The solver uses the L3 formula as a global
non-Bachelier seed.  After reducing the input to a normalised out-of-the-money
call price $c$ and total volatility $v=\sigma\sqrt{T}$, it refines the
smaller-tail log objective: $\ln(c)$ for $c\le1/2$ and $\ln(1-c)$ for
$c>1/2$.  At an iterate $v_n$, write
\[
  \eta_n=-\frac{F(v_n)}{F'(v_n)},\qquad
  \lambda_n=\frac{F(v_n)F''(v_n)}{[F'(v_n)]^2}.
\]
On the lower tail,
$F(v)=\ln(c(x,v))-\ln(c_{\mathrm{target}})$ and ThiopheneIV uses the
Euler--Chebyshev method of order three~\cite[p.~81]{traub1982},
\[
  v_{n+1}=v_n+\eta_n\left(1+\frac{\lambda_n}{2}\right).
\]
On the upper tail,
$F(v)=\ln(1-c_{\mathrm{target}})-\ln(1-c(x,v))$ and ThiopheneIV uses Halley's
method in the closely related rational form
\[
  v_{n+1}=v_n+\frac{\eta_n}{1-\lambda_n/2}.
\]
Starting from the Choi L3 lower-bound seed, both maps are monotone increasing
and bounded above by the true root in exact arithmetic
(Appendices~\ref{app:chebyshev-convergence} and
\ref{app:loggap-halley-convergence}).

The paper contributes an explicit normalised L3-seeded algorithm, monotone
no-overshoot proofs for its lower-tail Euler--Chebyshev and upper-tail Halley
maps, and the production boundary handling needed for robust double-precision
use.  Accuracy and latency are compared against multiprecision Black references
and recent implied-volatility solvers, including J\"ackel's reference
implementation.
\section{Problem Formulation}
\label{sec:problem}

\subsection{Normalization}
\label{sec:normalization}

All prices are first reduced to the same normalised OTM representation.
Let $C$ and $P$ denote undiscounted call and put prices with forward
$F$, strike $K$, and time to expiry~$T$.  We call an option at the money
(ATM) when $F=K$, and near ATM when $|\ln\left(F/K\right)|$ is small.  A call is
out of the money (OTM) when $F<K$, while a put is OTM when $F>K$; the
opposite cases are in the money (ITM).  Rather than carrying separate
call, put, ITM, and OTM formulae through the solver, we order the
forward--strike pair as
\[
  F_* = \min(F,K), \qquad K_* = \max(F,K), \qquad
  \frac{F_*}{K_*}=\min(F/K,K/F) \le 1.
\]
The corresponding OTM value is obtained by put--call parity, and by
exchanging the forward and strike for OTM puts:
\[
C_{\mathrm{OTM}}=
\begin{cases}
  C, & \text{call and } F\le K,\\
  C-(F-K), & \text{call and } F>K,\\
  P, & \text{put and } F>K,\\
  P-(K-F), & \text{put and } F<K.
\end{cases}
\]
Thus every admissible input is represented as an undiscounted OTM call
on $(F_*,K_*)$.  The variables passed to the scalar inversion are
\begin{linenomath}
\begin{equation}\label{eq:normalise}
  x = \ln\left(\frac{F_*}{K_*}\right) \le 0, \qquad
  e^x = \frac{F_*}{K_*}, \qquad
  c = \frac{C_{\mathrm{OTM}}}{F_*}, \qquad
  v = \sigma\sqrt{T}.
\end{equation}
\end{linenomath}

We also use the sqrt-forward normalised price
\begin{linenomath}
\begin{equation}\label{eq:beta-normalise}
  \beta = c e^{x/2} = \frac{C_{\mathrm{OTM}}}{\sqrt{F_*K_*}},
  \qquad c=\beta e^{-x/2}.
\end{equation}
\end{linenomath}
This is only a change of price scale; the solver input remains the
OTM-forward-normalised price~$c$.

This normalisation removes intrinsic value before inversion.  In
particular, originally ITM options are priced through their OTM parity
legs, avoiding the cancellation that would occur if one subtracted a
large intrinsic component from the original option price.

Here and throughout, $\Phi$ denotes the standard normal cumulative
distribution function.

In these coordinates the Black formula for the normalised OTM call price is
\begin{linenomath}
\begin{equation}\label{eq:black}
  c(x, v) = \Phi\left(\frac{x}{v} + \frac{v}{2}\right)
    - e^{-x}\Phi\left(\frac{x}{v} - \frac{v}{2}\right).
\end{equation}
\end{linenomath}

\subsection{Tail-Log Objective and erfcx Decomposition}
\label{sec:log-price-objective}
Following J\"ackel~\cite{jaeckel2017} and
Choi et~al.~\cite{choi2023}, the iteration is performed in logarithmic price
space rather than on the raw price residual.  For a target normalised
price $c_{\mathrm{target}}$, ThiopheneIV uses the smaller tail
\begin{linenomath}
\begin{equation}\label{eq:objective}
  f(v)=
  \begin{cases}
    \ln\left(c(x,v)\right)-\ln\left(c_{\mathrm{target}}\right),
      & c_{\mathrm{target}}\le 1/2,\\[3pt]
    \ln\left(1-c(x,v)\right)-\ln\left(1-c_{\mathrm{target}}\right),
      & c_{\mathrm{target}}>1/2.
  \end{cases}
\end{equation}
\end{linenomath}

\begin{Proposition}\label{prop:logc}
For $h = x/v$ and $t = v/2$, the log-price admits the
representation
\begin{equation}\label{eq:logc}
  \ln\left(c\right) = -\tfrac{1}{2}(h^2 + t^2) - \ln\left(2\right) - \tfrac{x}{2}
           + \ln\left(N^+ - N^-\right),
\end{equation}
where
\begin{equation}\label{eq:erfcx}
  N^+ = \erfcx\left(-(h+t)/\sqrt{2}\right), \qquad
  N^- = \erfcx\left(-(h-t)/\sqrt{2}\right).
\end{equation}
\end{Proposition}

\begin{proof}
Using~\eqref{eq:black} and
$\Phi\left(z\right) = \tfrac{1}{2}\erfc\left(-z/\sqrt{2}\right)$ gives
\[
  c = \tfrac{1}{2}\bigl[
    \erfc\left(-(h+t)/\sqrt{2}\right)
    - e^{-x}\,\erfc\left(-(h-t)/\sqrt{2}\right)
  \bigr].
\]
Now write $\erfc\left(z\right) = e^{-z^2}\erfcx\left(z\right)$.  Since $ht=x/2$,
\[
  e^{-(h+t)^2/2}=e^{-(h^2+t^2)/2-x/2},
  \qquad
  e^{-x}e^{-(h-t)^2/2}=e^{-(h^2+t^2)/2-x/2}.
\]
The two terms therefore share the common factor
$e^{-(h^2+t^2)/2-x/2}/2$.  Factoring it out and taking logarithms
gives~\eqref{eq:logc}.
\end{proof}

\begin{Proposition}\label{prop:loggap}
For the same $h=x/v$ and $t=v/2$, the complementary log-price admits the
representation
\begin{linenomath}
\begin{equation}\label{eq:loggap}
  \ln\left(1-c\right)
  = -\tfrac{1}{2}(h+t)^2 - \ln\left(2\right)
    + \ln\left(M^+ + M^-\right),
\end{equation}
\end{linenomath}
where
\begin{equation}\label{eq:erfcx-gap}
  M^+ = \erfcx\left((h+t)/\sqrt{2}\right), \qquad
  M^- = \erfcx\left(-(h-t)/\sqrt{2}\right).
\end{equation}
\end{Proposition}

\begin{proof}
Using~\eqref{eq:black},
\[
  1-c = \tfrac{1}{2}\erfc\left((h+t)/\sqrt{2}\right)
      + \tfrac{1}{2}e^{-x}\erfc\left(-(h-t)/\sqrt{2}\right).
\]
After writing $\erfc(z)=e^{-z^2}\erfcx(z)$, the two terms share the common
factor $e^{-(h+t)^2/2}/2$, because $ht=x/2$ implies
$e^{-x}e^{-(h-t)^2/2}=e^{-(h+t)^2/2}$.  Factoring it out and taking logarithms
gives~\eqref{eq:loggap}.
\end{proof}

Equation~\eqref{eq:loggap} is used whenever $c_{\mathrm{target}}>1/2$.  It avoids
the loss of scale in $\ln(c)$ near the upper price bound while requiring the
same two $\erfcx$ evaluations as~\eqref{eq:logc}.

\begin{Remark}
The scaled complementary error function
$\erfcx\left(z\right) = e^{z^2}\erfc\left(z\right)$ is bounded and positive for all
real~$z$.  Consequently, the difference $N^+ - N^-$ remains computable
even when the price itself would underflow to zero in double precision.
This is the main reason for using the log-price formulation: the
objective~$f(v)$ remains well defined and smoothly differentiable for
arbitrarily deep OTM options.
\end{Remark}

Equations~\eqref{eq:logc} and~\eqref{eq:loggap} are the default ThiopheneIV
iteration objectives: the unpolished solver selects the smaller logarithmic
tail, applying Euler--Chebyshev steps on $\ln(c)$ and Halley steps on
$\ln(1-c)$.  The expanded J\"ackel branch in
Appendix~\ref{app:expanded-black} is used separately as the benchmark reference
price and as the optional final-polish target.

The same numerical choices also matter before inversion, when prices are
produced or used as targets.  Table~\ref{tab:pricing-diagnostics} compares
three double-precision pricing paths against a 512-bit reference price: the
textbook normal-CDF formula~\eqref{eq:black}, the beta-normalised $\erfcx$/log
formula~\eqref{eq:logc}, and the expanded J\"ackel reference branch of
Appendix~\ref{app:expanded-black}. 

The CDF path uses a normal CDF implemented
through Apache Commons Numbers \texttt{erfc}; the $\erfcx$/log path evaluates
the sqrt-forward-normalised price $\beta=c\sqrt{e^x}$ with Apache Commons
Numbers \texttt{erfcx} and then divides by $\sqrt{e^x}$; and the expanded path
uses the J\"ackel-style branch.  The plotting script draws these Java-benchmark values and
supplies the high-precision error scale.  The examples are concrete OTM
equity-style quotes, written in terms of the forward moneyness $F_*/K_*$,
expiry, annualised volatility, and total volatility $v=\sigma\sqrt{T}$.

The $\erfcx$ dependency should be a high-quality scaled complementary error
function, not the literal product $e^{z^2}\erfc(z)$ in the positive tail.
We found the implementations of Apache
Commons Numbers to be of the highest quality (the best accuracy). Cody's implementation from
Netlib is nearly as good, while Johnson's implementation part of \emph{Julia}'s \texttt{SpecialFunctions.jl} is less accurate, but may still be satisfactory (especially with the polishing step) for the double-precision regimes tested
here (see Appendix~\ref{sec:erfcx-impact}).  

\begin{table}[!htbp]
\caption{Concrete double-precision Black pricing diagnostics for the
OTM-normalised call price $c=C_{\mathrm{OTM}}/F_*$.  The scenario column gives
the market-style expiry/volatility description when one is used; the input
column always lists the starting $(F_*/K_*,v)$.  The reference price is shown at the starting $v$; the error
columns report the maximum relative error over 512 successive double-precision
values of $v$ starting there, relative to a high-precision reference price.\label{tab:pricing-diagnostics}}
\scriptsize
\setlength{\tabcolsep}{3pt}
\begin{tabularx}{\textwidth}{lcccccX}
\toprule
\textbf{Scenario} & $(F_*/K_*,v)$ & $c_{\mathrm{ref}}$ &
\textbf{CDF} & \textbf{erfcx/log} & \textbf{Expanded J\"ackel} & \textbf{Lesson} \\
\midrule
Two-week 7.5\% OTM, 15\% vol
& $(0.925,2.99{\times}10^{-2})$
& $4.42{\times}10^{-5}$
& $1.6{\times}10^{-13}$
& $5.2{\times}10^{-14}$
& $1.9{\times}10^{-15}$
& Tail cancellation is already visible in the CDF path; the branched reference is cleaner. \\
Monotonicity diagnostic, $v=5\%$
& $(0.955,5.00{\times}10^{-2})$
& $4.94{\times}10^{-3}$
& $2.7{\times}10^{-14}$
& $1.7{\times}10^{-14}$
& $7.1{\times}10^{-16}$
& This later figure case has small absolute pricing errors, but still tens to hundreds of price ulps. \\
One-day 0.5\% OTM, 10\% vol
& $(0.995,6.30{\times}10^{-3})$
& $7.65{\times}10^{-4}$
& $1.6{\times}10^{-13}$
& $1.1{\times}10^{-13}$
& $5.1{\times}10^{-16}$
& The beta-normalised erfcx path is good with Commons erfcx; the expanded branch gives the reference. \\
\bottomrule
\end{tabularx}
\end{table}

Figure~\ref{fig:pricing-formula-accuracy-diagnostic} visualises two of these
one-ulp sweeps in total volatility.  In these practical cases,
the beta-normalised $\erfcx$/log formula improves on the textbook CDF path, while the
branched reference remains cleaner still.

\begin{figure}[!htbp]
\centering
\includegraphics[width=\textwidth]{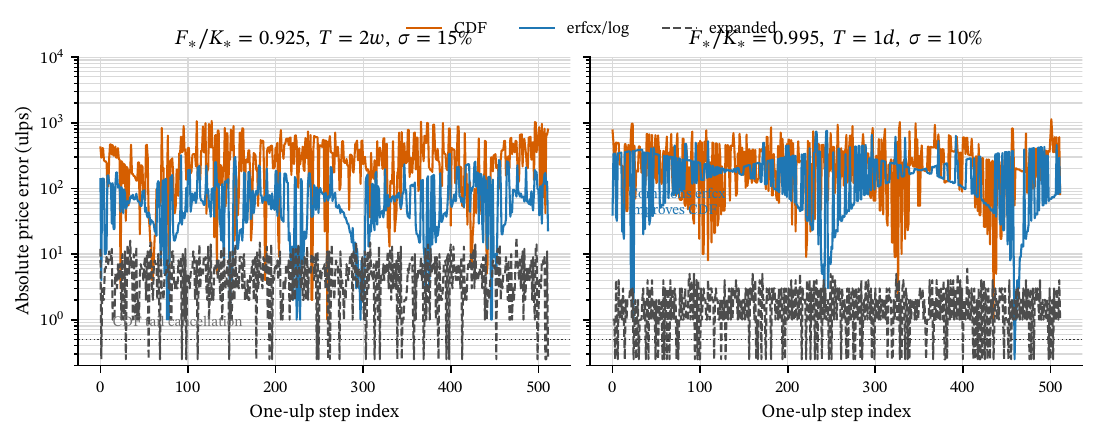}
\caption{Pricing-formula accuracy diagnostic against a high-precision reference.
Each panel advances $v$ by 512 successive double-precision values and plots
absolute price error in ulps of the high-precision OTM-normalised price.  The
left panel shows a realistic short-expiry OTM quote where the CDF formula loses
tail bits; the right panel shows a near-ATM short-expiry quote where a
high-quality $\erfcx$ implementation avoids the larger CDF error, but the
expanded branch is still the last-bit reference.
\label{fig:pricing-formula-accuracy-diagnostic}}
\end{figure}

Figure~\ref{fig:pricing-formula-monotonicity-diagnostic} separates the signed
price-step diagnostic because it is visually easier to read on its own and it
illustrates a different point.  Even when all three formulas are accurate in
absolute terms, the double-precision price sequence need not be monotone under
successive one-ulp increases of total volatility.  This happens for the CDF,
beta-normalised $\erfcx$/log, and expanded J\"ackel-style evaluations; higher precision
reduces cancellation but does not turn a rounded double-precision evaluation
into the real-valued monotone function.

\begin{figure}[!htbp]
\centering
\includegraphics[width=\textwidth]{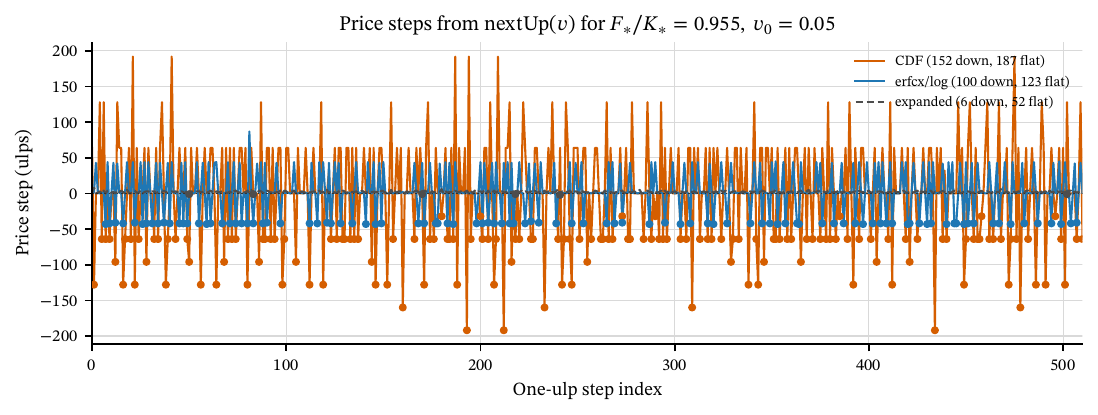}
\caption{One-ulp price-step diagnostic for $F_*/K_*=0.955$ and $v=0.05$.
Each curve shows the local OTM-price step produced by advancing total volatility
by one double-precision value.  Negative points mark local downward steps
despite the monotonicity of the real-valued Black price.
\label{fig:pricing-formula-monotonicity-diagnostic}}
\end{figure}

Table~\ref{tab:pricing-global} repeats the comparison on practical quoted-price
grids rather than irrelevant deep tails.  It confirms the same qualitative
pattern: the CDF path is limited by tail cancellation, the Commons $\erfcx$/log
path is usually cleaner in OTM regimes, and the expanded branch remains the
reference evaluation for pricing diagnostics and optional polish.

\begin{table}[!htbp]
\caption{Regime-level assessment of direct Black pricing formulas against a
200-decimal-digit reference on practical OTM quote grids.  The broad grid uses
$0.80\le F_*/K_*\le0.995$, expiries from one day to one year, and annualised
volatilities from 10\% to 120\%.  The near-ATM grid uses
$0.99\le F_*/K_*\le0.9995$, expiries from one day to two weeks, and annualised
volatilities from 5\% to 80\%.  Deep tails below $c=10^{-12}$ are omitted here.
\label{tab:pricing-global}}
\scriptsize
\setlength{\tabcolsep}{3pt}
\begin{tabularx}{\textwidth}{lccccX}
\toprule
\textbf{Regime} & \textbf{Points} & \textbf{CDF p99} & \textbf{erfcx/log p99} & \textbf{Expanded p99} & \textbf{Conclusion} \\
\midrule
Broad OTM grid, $10^{-8}\le c\le5{\times}10^{-2}$
& 390 & $3.9{\times}10^{-13}$ & $5.4{\times}10^{-14}$ & $2.0{\times}10^{-15}$
& The Commons-erfcx path reduces CDF cancellation, and the reference branch is cleaner. \\
Tiny quoted premiums, $10^{-12}\le c<10^{-8}$
& 19 & $2.1{\times}10^{-12}$ & $8.9{\times}10^{-14}$ & $5.9{\times}10^{-15}$
& Tail cancellation increasingly favours the log/expanded formulations. \\
Near-ATM short-dated grid, $c\le5{\times}10^{-2}$
& 186 & $2.8{\times}10^{-13}$ & $8.8{\times}10^{-14}$ & $4.6{\times}10^{-16}$
& Direct erfcx/log is useful, but the expanded branch is the last-bit reference. \\
\bottomrule
\end{tabularx}
\end{table}

This improvement should not be confused with a guarantee of strict monotonicity
in floating-point arithmetic.  The real Black price is monotone in total
volatility, but a particular double-precision evaluation need not be monotone
when the input is advanced by one ulp (unit in the last place, the spacing
between adjacent double-precision numbers at the value under discussion); this
distinction is highlighted in~\cite{lefloc2024monotonicity}.  The expansion is
valuable for reducing cancellation, but Figure~\ref{fig:pricing-formula-monotonicity-diagnostic}
shows that it should not be read as a floating-point monotonicity guarantee.


\section{Solver Design}
\label{sec:solver}

\subsection{Log-Price Derivatives}
\label{sec:derivatives}

ThiopheneIV refines total volatility with the selected tail-log objective from
Section~\ref{sec:log-price-objective}.  Let
\[
  g(v)=\ln\left(c(x,v)\right)-\ln\left(c_{\mathrm{target}}\right),
  \qquad h=x/v,
  \qquad t=v/2.
\]
The derivatives needed by the Euler--Chebyshev correction follow from the
same erfcx decomposition used to evaluate the objective.

\begin{Proposition}\label{prop:derivatives}
In the log-price formulation, the first three derivatives of
$g_0(v)=\ln\left(c(x,v)\right)$ admit closed-form expressions:
\begin{align}
  g_0'(v) &= \frac{2/\sqrt{2\pi}}{N^+ - N^-},
  \label{eq:logvega}\\[3pt]
  \frac{g_0''(v)}{g_0'(v)} &= \frac{(h+t)(h-t)}{v} - g_0'(v),
  \label{eq:d2}\\[3pt]
  \frac{g_0^{(3)}(v)}{g_0'(v)} &=
    \frac{-3h^2 - t^2 + (h^2-t^2)^2}{v^2}
    - 3\,g_0'(v)\,\frac{g_0''(v)}{g_0'(v)}
    - \bigl[g_0'(v)\bigr]^2.
  \label{eq:d3}
\end{align}
\end{Proposition}

\begin{proof}
The vega of the normalised Black formula is
$\partial c / \partial v = \phi\left(d_1\right)$, where
$\phi$ is the standard normal density.  Thus
$g_0'(v)=c^{-1}\phi\left(d_1\right)$.  Expressing
$\phi\left(d_1\right)=(2\pi)^{-1/2}e^{-d_1^2/2}$ and using
Equation~\eqref{eq:erfcx} gives~\eqref{eq:logvega};
Equations~\eqref{eq:d2} and~\eqref{eq:d3} follow by differentiating
with respect to~$v$ and simplifying.
\end{proof}

For the complementary branch, write $S=M^++M^-$ and
$\ell_q(v)=\ln(1-c(x,v))$.  The common exponential factor in
\eqref{eq:loggap} cancels from the Newton and second-derivative ratios, giving
\begin{align}
  \ell_q'(v) &= -\frac{2}{\sqrt{2\pi}\,S},
  \label{eq:loggap-vega}\\[3pt]
  \frac{\ell_q''(v)}{\ell_q'(v)}
  &= -(h+t)\left(\frac{1}{2}-\frac{x}{v^2}\right)
     + \frac{2}{\sqrt{2\pi}\,S}.
  \label{eq:loggap-d2}
\end{align}
Thus the upper-tail Halley step needs the same two $\erfcx$ values and one
logarithm as the direct log-price step; no additional exponential evaluation is
required.

Once the two $\erfcx$ values for the selected tail are available from the
objective evaluation, the derivative ratios involve only elementary arithmetic.
The Chebyshev and Halley corrections therefore cost little more than a Newton
step with the same $\erfcx$ calls, but have cubic local convergence at the
simple root.

\subsection{Choi--Huh--Su L3 Seed}
\label{sec:choi-seed}

For $x\le0$, set $k=-x\ge0$ and $E=e^k$.  The Choi--Huh--Su L3 seed maps the
normalised OTM price to
\begin{linenomath}
\begin{equation}\label{eq:choi-l3-main}
  p_3=\frac{c(c+E)}{2c+E-1},\qquad z_3=\Phi^{-1}\left(p_3\right),
\end{equation}
\end{linenomath}
and then solves $z_3=d_1(v)=-k/v+v/2$ for a positive total volatility.  The
root is evaluated in rationalised form when $z_3<0$ to avoid cancellation;
Appendix~\ref{app:choi-l3} gives the exact transcription, including the ATM
limit.

The seed is not selected because it is the most accurate raw approximation.
Its advantage is structural: Choi et~al.~\cite{choi2023} prove that this value
is a lower bound for the admissible implied volatility.  ThiopheneIV uses that
lower-side property as the entry point for a monotone refinement rather than as
a merely empirical basin guess.  The inverse normal CDF is evaluated by a
piecewise rational approximation; subsequent cubic steps and the optional
J\"ackel--Newton polish absorb the last-bit
differences introduced by that approximation.

\subsection{Cubic Refinement}
\label{sec:chebyshev}

Given a finite positive seed $v_n$ below the root, define the Newton
displacement and the usual dimensionless Chebyshev--Halley curvature parameter
\[
  \eta_n=-\frac{g(v_n)}{g'(v_n)},\qquad
  \lambda_n=\frac{g(v_n)g''(v_n)}{[g'(v_n)]^2}.
\]
On the lower-tail log-price branch, ThiopheneIV applies the Euler--Chebyshev
method of order three~\cite[p.~81]{traub1982} in the compact form
\begin{linenomath}
\begin{equation}\label{eq:cheb-main}
  v_{n+1}=v_n+\eta_n\left(1+\frac{1}{2}\lambda_n\right).
\end{equation}
\end{linenomath}
The production path uses three full-precision refinement steps on either
objective branch.  Appendix~\ref{app:chebyshev-convergence}
proves that, in real arithmetic, the map in~\eqref{eq:cheb-main} satisfies
$0<T(v)-v\le v_*-v$ for every $v\in(0,v_*)$ when $x\le0$ on the lower-tail
log-price branch.  Therefore the corresponding sequence from the Choi L3
lower-bound seed is monotone increasing, bounded by the true root, and
convergent.  The upper-tail branch instead applies Halley's method to the
complementary objective, using the derivative ratios in
\eqref{eq:loggap-vega}--\eqref{eq:loggap-d2} to avoid loss of scale when
$c>1/2$:
\begin{linenomath}
\begin{equation}\label{eq:halley-main}
  v_{n+1}=v_n+\frac{\eta_n}{1-\frac{1}{2}\lambda_n},
\end{equation}
\end{linenomath}
where $\eta_n=-(\ell_q(v_n)-\ell_q(v_*))/\ell_q'(v_n)>0$ and
$\lambda_n=(\ell_q(v_n)-\ell_q(v_*))\ell_q''(v_n)/[\ell_q'(v_n)]^2$.
Appendix~\ref{app:loggap-halley-convergence} records the corresponding
monotone-convergence result.

\subsection{Optional J\"ackel--Newton Polish}
\label{sec:nj-polish}

The default ThiopheneIV objective is the selected erfcx/tail-log formula.  When the goal
is instead to match the expanded J\"ackel reference price as closely as
possible, one final Newton correction is applied on the lower-price half
$c\le1/2$ to the residual in the sqrt-forward normalised price
$\beta=c e^{x/2}$ from~\eqref{eq:beta-normalise}.  For $c>1/2$, the core solver is already using the complementary
$\ln(1-c)$ objective, and a final direct-price residual can be less well
conditioned than the quantity being inverted.  The correction uses the reference
price from Appendix~\ref{app:expanded-black}, converted to the same
sqrt-forward normalisation, and the analytic normalised vega.
The cutoff $c\le1/2$ is an objective-conditioning cutoff: above the midpoint the
small quantity is $1-c$, so the complementary log objective remains better
conditioned than a direct price residual.  Section~\ref{sec:polish-diagnostic}
separates this cutoff from J\"ackel's expanded price-evaluation regions and
measures the numerical effect of moving it.

\subsection{Guards and Hot Path}
\label{sec:guards}

Figure~\ref{fig:thiophene-branch-flow} summarises the input-domain branch map.
The guards detailed in Appendix~\ref{sec:guards_and_branches} handle adversarial inputs rather than only smooth benchmark
grids.
\begin{figure}[!htbp]
\centering
\includegraphics[width=\textwidth]{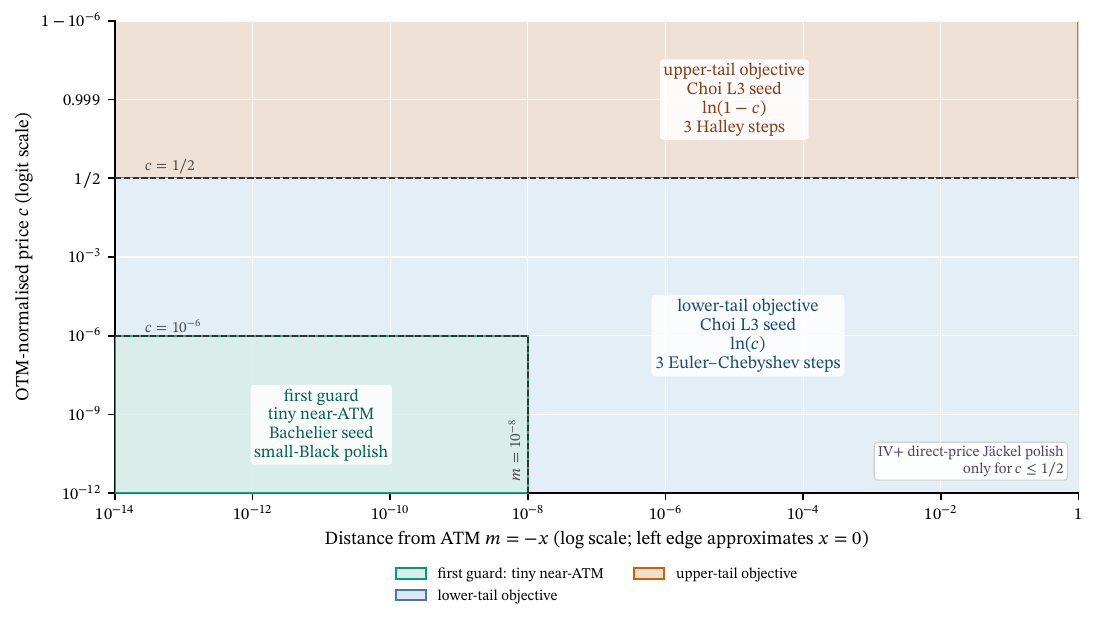}
\caption{ThiopheneIV input-domain branch map after OTM normalisation, shown in
the variables $m=-x$ and $c=C_{\mathrm{OTM}}/F_*$.  The vertical axis is a logit
scale, so the top tick is capped at $1-10^{-6}$ only to keep the approach to
the upper bound $c=1$ finite.  The tiny near-ATM guard is tested first.  For
all remaining inputs the solver forms a Choi L3 seed and repairs it only if it
is invalid or non-finite.  The unpolished route then uses three
Euler--Chebyshev steps on $\ln(c)$ for $c\le1/2$ and three Halley steps on
$\ln(1-c)$ for $c>1/2$.  The optional ThiopheneIV+ price polish is applied only
for $c\le1/2$; for $c>1/2$ the complementary-gap objective is retained without
an expanded J\"ackel price polish.
\label{fig:thiophene-branch-flow}}
\end{figure}

\section{Experimental Setup}
\label{sec:setup}

\subsection{Reference Price Computation}
\label{sec:refprice}

Benchmarking a solver to machine precision requires reference prices accurate
to the last bit.  The textbook Black--Scholes formula
$c = \Phi\left(d_1\right) - e^{-x} \Phi\left(d_2\right)$ suffers from catastrophic
cancellation for deep OTM options, and even the erfcx/log objective used by the
unpolished solver can lose relative bits in nearly-ATM, very-low-volatility
cases because it subtracts two close erfcx values.  Therefore the main accuracy
tables use prices generated from a multiprecision Black evaluator at the known
reference total volatility.  The resulting OTM-normalised prices are rounded to
double precision and stored with the source distribution, so the benchmark does
not regenerate a double-precision pricing oracle on each run.  For each
reference total volatility, the stored price is passed to the solver in the same
$c=C_{\mathrm{OTM}}/F_*$ convention as any other input.  Thus the reported
accuracy is a comparison against rounded multiprecision Black prices, not merely
against the internal objective of the unpolished solver.
ITM prices are recovered via put--call parity (a well-conditioned
addition of intrinsic value), never by direct evaluation.

\subsection{Test Datasets}
\label{sec:datasets}

To ensure robustness, we benchmark on eight datasets spanning the
full range of practically relevant option parameters
(Table~\ref{tab:datasets}).

\begin{table}[!htbp]
\caption{Test datasets.\label{tab:datasets}}
\begin{tabularx}{\textwidth}{lrX}
\toprule
\textbf{Dataset} & \textbf{Cases} & \textbf{Description} \\
\midrule
CLY-3D    & \num{51321} & Cui--Liu--Yao three-dimensional grid:
            $K \in [105,800]$, $T \in [0.01,2]$,
            $\sigma \in [0.01,0.99]$, filtered at price $10^{-20}$ \\
CLY-20    & \num{1600} & Cui--Liu--Yao fixed-volatility surface:
            $\sigma=20\%$, $K \in [105,180]$, $T \in [0.1,2]$ \\
CLY-80    & \num{1600} & Cui--Liu--Yao fixed-volatility surface:
            $\sigma=80\%$, $K \in [105,800]$, $T \in [0.1,2]$ \\
J\"ackel  & \num{5182}  & Wide moneyness: $K/F \in [0.5, 8]$,
            $\sigma$ up to 4.0 \\
Market    & \num{7151}  & Realistic: $K/F \in [0.7, 1.5]$,
            $T$ from 1/252 to 5\,yr \\
Corners   & 278         & Edge cases: low-vol/short-mat,
            high-vol/deep-OTM, near-ATM small-price, and saturated upper-price cases \\
Stress    & \num{1270}  & Extremes: $K/F$ up to $100\times$,
            $T \in [0.001, 10]$ \\
HighVol   & 149         & Fallback stress zone: $|x| \ge 3$,
            $c \in (0.05, 0.95)$, $\sigma$ up to 2.5 \\
\bottomrule
\end{tabularx}
\end{table}

The first three datasets reproduce the comparison grids from
Cui--Liu--Yao~\cite{cui2025}: the main three-dimensional grid and the
two fixed-volatility surface grids used in their numerical comparison
with previous literature.  The HighVol dataset
specifically targets the region $|x| \ge 3$ with large option prices; cases are
filtered to $c < 0.95$ to reflect the practical use of put--call
parity for deep ITM options.
Appendix~\ref{app:benchmark-construction} gives the exact construction
rules used by the benchmark generator.

\subsection{Solvers Under Comparison}
\label{sec:solvers}

The main comparison is between ThiopheneIV and J\"ackel's \emph{Let's Be
Rational}.  The timing rows use the same benchmark protocol:
minimum of 500 sweeps and 3 independent runs.

\begin{itemize}
  \item \textbf{J\"ackel's \emph{Let's Be Rational}~\cite{jaeckel2017}.}
        Region-dependent asymptotic expansions, log-space iteration, and a
      complementary objective.  The comparison uses the normalised API with
        $\beta=c\sqrt{F/K}$ and the original two-iteration default.
  \item \textbf{ThiopheneIV.}  Choi--Huh--Su L3 seed, three full-precision
      lower-tail Euler--Chebyshev steps, upper-tail Halley steps, and production guards.
        We denote the polished configuration by
        \textbf{ThiopheneIV+}: ThiopheneIV with one final expanded J\"ackel price
        polish on $c\le1/2$ only; the upper half remains on the
        complementary-gap objective.
\end{itemize}

\FloatBarrier
\section{Results}
\label{sec:results}

\subsection{Accuracy and Timing}

Table~\ref{tab:accuracy} reports accuracy and latency on the eight benchmark
datasets.  Input prices for the accuracy block are generated from a
multiprecision Black price at the known total volatility
$v_{\rm ref}$, rounded to double precision, and then passed to the solver as OTM-normalised prices.
Errors are measured in ulps of the reference total volatility.  For each test
case, if the reference total volatility is $v_{\rm ref}$ and the solver returns
$\hat v=\hat\sigma\sqrt{T}$, the reported per-case error is
\[
  \frac{|\hat v-v_{\rm ref}|}{\operatorname{nextUp}(v_{\rm ref})-v_{\rm ref}}.
\]
Thus the unit is the local double-precision spacing at the reference volatility, not a
fixed decimal tolerance.  For example, at $v_{\rm ref}=0.20$, one ulp is
$2.7755575615628914\times10^{-17}$; an 8-ulp error corresponds to an absolute
total-volatility error of about $2.22\times10^{-16}$, or a relative error of
about $1.11\times10^{-15}$.  Table~\ref{tab:accuracy} reports the maximum of
this quantity over each dataset.  This makes the accuracy table a comparison
against a multiprecision price-generation oracle rather than against the simpler
erfcx/log formula.

\begin{table}[!htbp]
\caption{Accuracy against rounded multiprecision Black reference prices and latency
by dataset.\label{tab:accuracy}}
\scriptsize
\setlength{\tabcolsep}{4pt}
\begin{tabularx}{\textwidth}{lCCCCCCCC}
\toprule
& \textbf{CLY-3D} & \textbf{CLY-20} & \textbf{CLY-80}
& \textbf{J\"ackel} & \textbf{Market} & \textbf{Corners}
& \textbf{Stress} & \textbf{HighVol} \\
\midrule
\multicolumn{9}{l}{\textbf{Accuracy -- max error (ulp of reference total volatility)}} \\[2pt]
J\"ackel    &  23 &  5 & 4 &  25 &  29 & 240 &  33 & 2 \\
ThiopheneIV & 133 & 62 & 7 &  89 & 177 & 329 & 138 & 2 \\
ThiopheneIV+  &  24 &  5 & 5 &  13 &  29 &  41 &  33 & 2 \\[4pt]
\multicolumn{9}{l}{\textbf{ThiopheneIV+ accuracy -- max absolute total-volatility error}} \\[2pt]
ThiopheneIV+
  & $7.8{\times}10^{-16}$ & $1.7{\times}10^{-16}$
  & $5.6{\times}10^{-16}$ & $6.2{\times}10^{-15}$
  & $8.9{\times}10^{-16}$ & $2.7{\times}10^{-15}$
  & $8.9{\times}10^{-16}$ & $8.9{\times}10^{-16}$ \\[4pt]
\multicolumn{9}{l}{\textbf{Latency (ns/call)}} \\[2pt]
J\"ackel    & 244 & 230 & 221 & 225 & 204 & 263 & 239 & 208 \\
ThiopheneIV & 167 & 164 & 165 & 171 & 170 & 177 & 176 & 166 \\
ThiopheneIV+  & 211 & 209 & 212 & 209 & 215 & 226 & 223 & 201 \\
\bottomrule
\end{tabularx}
\end{table}

Unpolished ThiopheneIV is faster than J\"ackel's solver on all
eight datasets in this run.  The optional Newton polish is still competitive with
J\"ackel's two-iteration reference path, but it is no longer a pure low-latency
configuration; its purpose is reference-price alignment on the lower-price half.
When the comparison target is instead the rounded multiprecision Black price,
individual polished cases can move by a few ulps in either direction once the
three fixed iterations have already reached the last-bit regime.
The upper half is left on the complementary-gap objective, avoiding the direct
price-residual conditioning illustrated in Appendix~\ref{app:worked-examples}.
The Corners ulp count is driven by a near-zero-vega case; the absolute-error row
shows that the ThiopheneIV+ maximum there is $2.7\times10^{-15}$ in total
volatility.

Figure~\ref{fig:jaeckel-comparison-slice} gives fixed-log-moneyness slices
against J\"ackel's normalised solver.  It checks that ThiopheneIV remains in the
same accuracy band as the J\"ackel reference on the diagnostic slices, and shows
how the optional ThiopheneIV+ polish collapses most of the lower-price
reference-price residual.

\begin{figure}[!htbp]
\centering
\includegraphics[width=\textwidth]{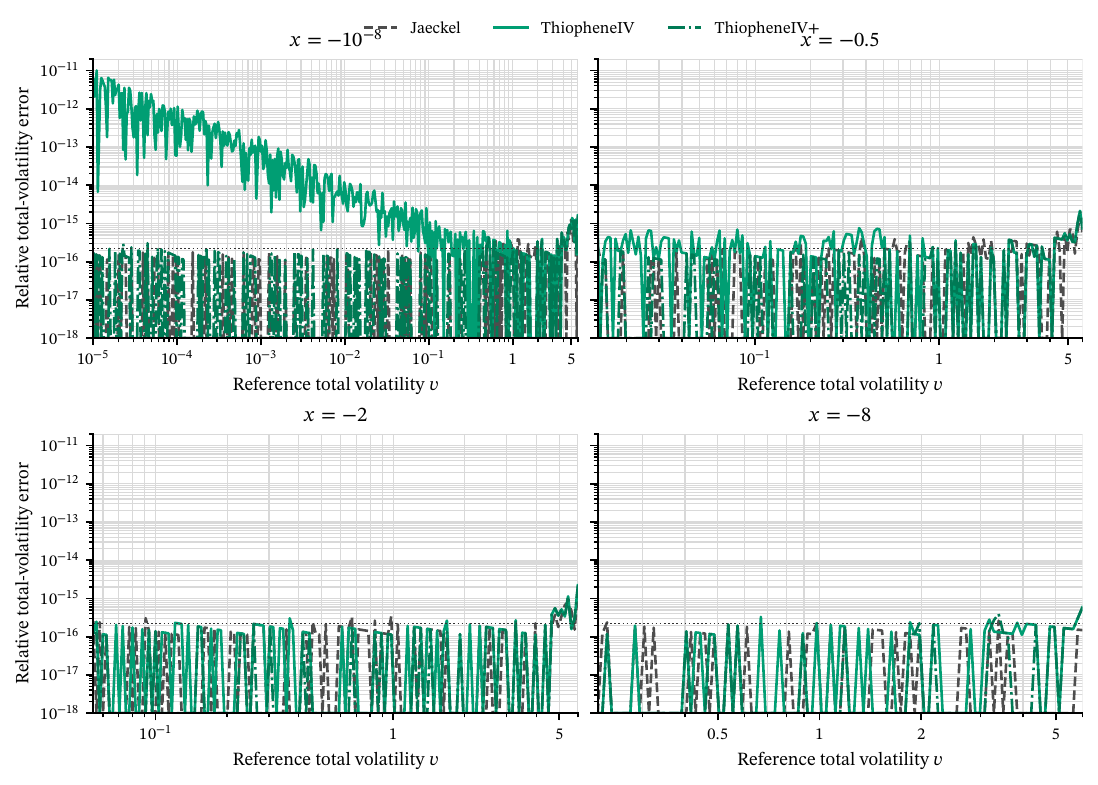}
\caption{Fixed-$x$ comparison of ThiopheneIV and ThiopheneIV+ against J\"ackel's
implementation.  Each panel varies the reference total volatility at fixed
log-moneyness and reports relative total-volatility error.
\label{fig:jaeckel-comparison-slice}}
\end{figure}

The separation between ThiopheneIV and ThiopheneIV+ in the first panel
($x=-10^{-8}$, $v\lesssim10^{-3}$) is the same kind of pricing-formula limit
highlighted by Figure~\ref{fig:pricing-formula-accuracy-diagnostic}, not a sign that
the Choi seed or the three cubic refinement steps have failed to converge.  The
input prices in the figure are generated by the expanded J\"ackel reference price.
Unpolished ThiopheneIV then solves the erfcx/log objective, whose two erfcx
terms are nearly equal in this near-ATM, very-small-total-variance corner.  At
this scale the erfcx/log formula has already lost enough relative bits that its
root is slightly different from the expanded-J\"ackel-reference root.  The
resulting relative price mismatch is typically $10^{-10}$--$10^{-12}$ in this
panel, which appears as about $10^{-12}$ relative volatility error.  J\"ackel and
ThiopheneIV+ both finish by matching the expanded J\"ackel reference price, so
their curves lie near the double-precision floor.

Figure~\ref{fig:thiophene-fixed-x-convergence} resolves the same style of
fixed-$x$ slices by iteration stage.  The Choi L3 seed starts on the lower side
of the root, and the lower-tail Chebyshev / upper-tail Halley steps compress the error to the
double-precision floor across the diagnostic grid.

\begin{figure}[!htbp]
\centering
\includegraphics[width=\textwidth]{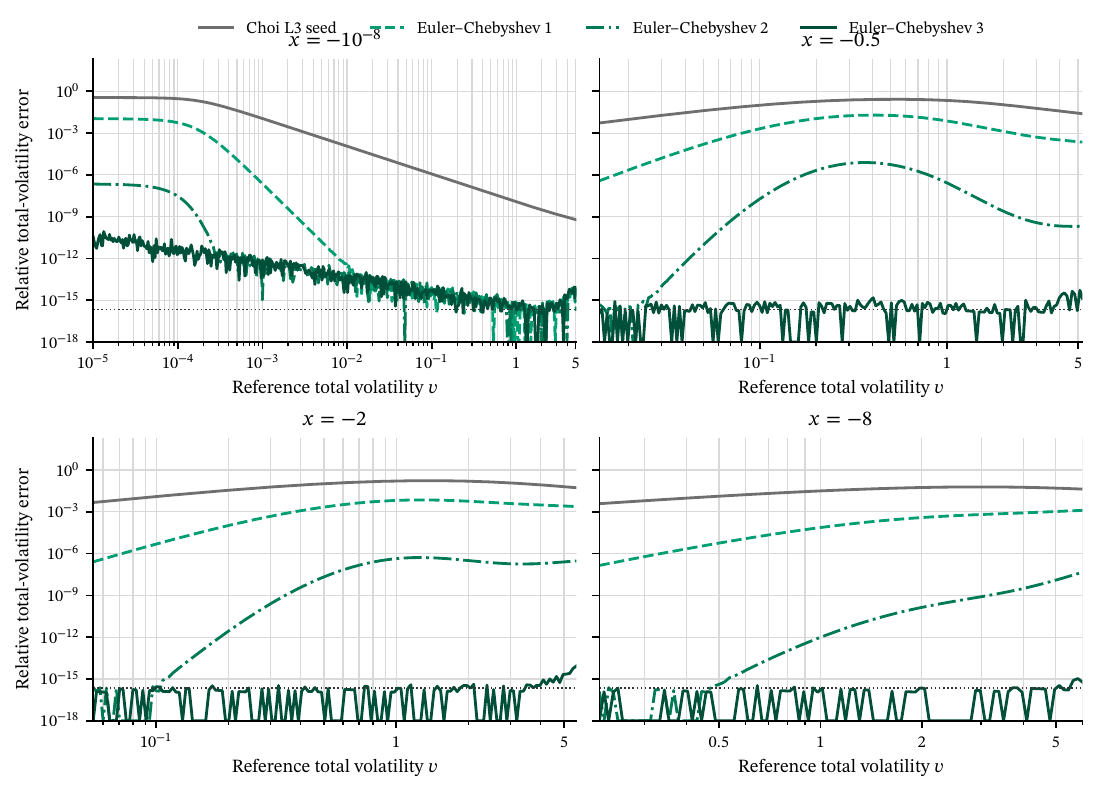}
\caption{Fixed-$x$ convergence slices for the Choi L3 seed and the
ThiopheneIV refinement chain.  Each panel fixes log-moneyness and varies the
reference total volatility.  Curves show relative total-volatility error after
the seed and after one, two, and three full-precision refinement steps;
values are floored at $10^{-18}$ only for log-scale plotting.
\label{fig:thiophene-fixed-x-convergence}}
\end{figure}

\subsection{J\"ackel--Newton Polish Diagnostic}
\label{sec:polish-diagnostic}

The cutoff $c\le1/2$ used by ThiopheneIV+ is not the boundary of J\"ackel's
expanded price-evaluation regions.  Let $c_{\mathrm{J}}$ denote the J\"ackel
expansion threshold derived in Appendix~\ref{app:expanded-black}, namely
\[
  c_{\mathrm{J}} = \erf\left(\sqrt{2}\,\epsilon^{1/16}\right)
  \approx 0.16650723223355586 .
\]
For $c\le c_{\mathrm{J}}$, the polish targets the expanded lower-tail J\"ackel
evaluator and reduces the regular-grid inverse errors to the last few ulps.
The intermediate band $c_{\mathrm{J}}<c\le1/2$ is different: the same
J\"ackel-style evaluator has already fallen back to the Cody erfc/erfcx formula.
Here $x$ is fixed and the sqrt-forward normalised price
$\beta=c e^{x/2}$ from~\eqref{eq:beta-normalise} is proportional to $c$, so
the Newton residual is still formed on the smaller-tail price; the correction
remains well conditioned, but its value depends on the price oracle being
matched.

\begin{table}[!htbp]
\caption{Effect of moving the optional J\"ackel--Newton polish cutoff.  The
vectors in the first two rows are maximum total-volatility errors, in the
dataset order of Table~\ref{tab:accuracy}.\label{tab:polish-threshold}}
\scriptsize
\setlength{\tabcolsep}{3pt}
{\renewcommand{\arraystretch}{1.15}
\begin{tabularx}{\textwidth}{>{\raggedright\arraybackslash}p{0.24\textwidth}>{\raggedright\arraybackslash}p{0.17\textwidth}>{\raggedright\arraybackslash}p{0.32\textwidth}>{\raggedright\arraybackslash}X}
\toprule
\textbf{Experiment} & \textbf{Target price} & \textbf{Numerical result} & \textbf{Takeaway} \\
\midrule
All benchmark cases; no final polish
& Rounded multiprecision
& Max total-volatility error: $133,62,7,89,177,329,138,2$ ulps
& Baseline after the three fixed refinement steps. \\
All benchmark cases; polish restricted to $c\le c_{\mathrm{J}}$
& Rounded multiprecision
& Max total-volatility error: $24,5,5,13,29,41,33,2$ ulps
& Same maxima as the $c\le1/2$ polish; the low-tail expansion region supplies the grid-level gain. \\
$c_{\mathrm{J}}<c\le1/2$ band (6603 cases); add polish in this band
& Rounded multiprecision
& Price residual better/worse/same: 2124/676/3803; largest gain/loss: 8/5 price ulps
& Mixed last-bit effect against the multiprecision oracle. \\
$c_{\mathrm{J}}<c\le1/2$ band (6603 cases); add polish in this band
& J\"ackel formula
& Price residual better/worse/same: 2188/121/4294; max total-volatility error: 9 to 6 ulps
& More favourable when the same J\"ackel formula supplies the target. \\
\bottomrule
\end{tabularx}}
\end{table}

Thus the polish above $c_{\mathrm{J}}$ is mainly a choice to align with
J\"ackel's double-precision pricing formula.  That formula is a high-quality
double oracle, but its rounded price can still differ in the last bit, or by a
few ulps after inversion, from a price generated in multiprecision and rounded
to double.

It is deliberately optional: when invoked it improves ulp-level agreement with
the expanded J\"ackel reference price, but it adds roughly 48--55\,ns per call.
This last correction is useful only when the surrounding system also treats that
reference price as the target.  Ideally, production systems would use the same
high-accuracy J\"ackel-style Black price for pricing, calibration, and inversion.
In practice, the straightforward Black--Scholes formula built from a library
normal CDF is extremely common; in such systems the pricing formula itself is
typically less accurate than the expanded J\"ackel reference price, so the polish
may provide no practical benefit.

Figure~\ref{fig:thiophene-polish-roundtrip-diagnostic} shows locally the same
effect reported in Table~\ref{tab:accuracy}: where the optional polish is active,
it reduces the residual of the expanded J\"ackel reference price, with the
incremental cost visible in the latency block of the table.

\begin{figure}[!htbp]
\centering
\includegraphics[width=\textwidth]{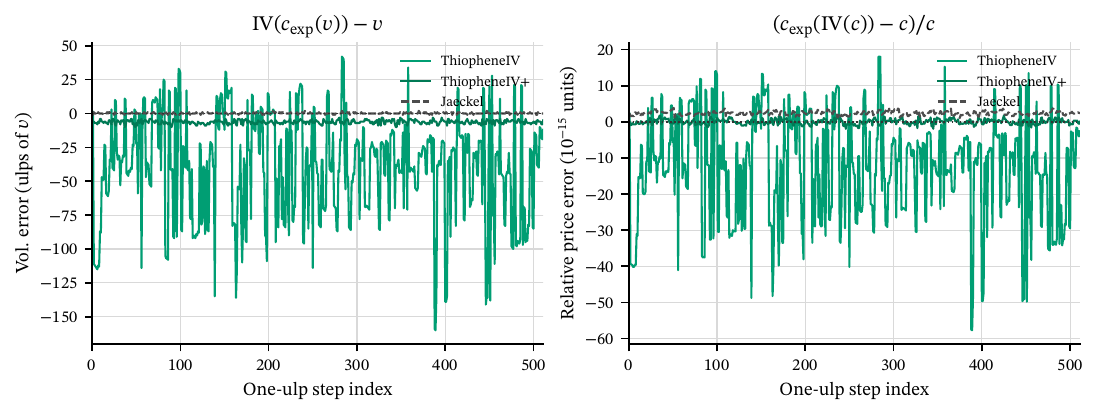}
\caption{Local round-trip diagnostics for the optional J\"ackel--Newton polish.
The left panel measures implied-volatility error when the input price is
generated by the expanded J\"ackel reference price; the right panel measures
the relative price residual after inversion and repricing with the same
reference formula.
\label{fig:thiophene-polish-roundtrip-diagnostic}}
\end{figure}

The small non-zero mean of the ThiopheneIV+ curve in the left panel is a
floating-point effect, not evidence that the polished method remains a lower bound.
After the final Newton polish, the sign is dominated by rounding, objective
evaluation mismatch, and ulp quantisation; vega maps the remaining few-volatility-ulp
bias into sub-ulp price changes, so the repriced residual can look mean-zero.

\subsection{Operation Costs}

\begin{table}[!htbp]
\caption{Micro-benchmarked operation costs relevant to ThiopheneIV in the
implementation.\label{tab:costs}}
\begin{tabularx}{\textwidth}{lC}
\toprule
\textbf{Operation} & \textbf{Cost (ns)} \\
\midrule
Thiophene L3 guess                  & 19.6 \\
1 exact erfcx (Boost/Commons)       & 5.9 \\
Natural logarithm                   & 4.9 \\
Exponential                         & 3.8 \\
Square root                         & 0.7 \\
Direct J\"ackel--Newton polish, when applied & about 48--55 \\
\bottomrule
\end{tabularx}
\end{table}

The Choi L3 seed is compact, but the three exact cubic refinement steps dominate
the hot path through erfcx and logarithm evaluations (Table \ref{tab:costs}).  J\"ackel--Newton polish is significant.

\section{Discussion}
\label{sec:discussion}

The proof-backed part of ThiopheneIV is deliberately narrow and explicit.  In
real arithmetic the Choi L3 seed starts on the lower side of the root, the
lower-tail Euler--Chebyshev map and upper-tail Halley map are increasing and do not
overshoot.  This gives a clean monotone-convergence story that is absent from a
purely empirical H3 chain.

The engineering story is less clean, and that is the important practical point.
A double-precision solver has to decide what to do when an inverse-normal probability is
rounded to the endpoint, when a candidate update is non-finite, when the price
is so close to the no-arbitrage upper bound that volatility is poorly
conditioned, or when a microscopic near-ATM price is better represented by a
Bachelier-limit expansion than by the Black tails themselves.  These branches do
not contradict the convergence theorem; they protect the preconditions under
which the theorem and the floating-point evaluation are useful.

J\"ackel's solver embodies a similar lesson from another direction.  Its
robustness comes from normalisation, asymptotic regions, complementary
objectives, and careful iteration choices.  ThiopheneIV replaces the central
seed-and-refine mechanism with Choi L3 plus monotone cubic refinement maps, but it
still needs comparable boundary awareness.  The optional J\"ackel--Newton polish
is another example: convergence to the erfcx/log objective is already obtained,
yet matching a different high-accuracy pricing objective to a few ulps can
require one more targeted correction on the lower-price half.  In the upper
half, the complementary-gap objective is the better-conditioned target.

The benchmark timings are implementation timings, not universal constants.  A
faster erfcx implementation, a vectorised batch path, or a different inverse
normal CDF would shift the absolute numbers.  The qualitative trade-off should
be more portable: a monotone seed-refine pair buys mathematical transparency,
while production robustness and final-reference matching add guard and polish
costs that must be budgeted explicitly.

\FloatBarrier
\section{Conclusions}
\label{sec:conclusion}

ThiopheneIV combines the Choi--Huh--Su L3 lower-bound seed with three
lower-tail Euler--Chebyshev corrections or, for prices above the midpoint,
three upper-tail Halley corrections on the complementary logarithmic objective.
In exact arithmetic, the Choi seed, lower-tail Euler--Chebyshev map, and
upper-tail Halley map give monotone convergence from below to the admissible
Black--Scholes implied volatility.  In the reported
benchmark, the unpolished implementation has low-double-digit to
few-hundred-ulp maximum errors against the multiprecision reference prices
and is faster than the J\"ackel comparison on all eight datasets.
With the optional J\"ackel--Newton polish on $c\le1/2$, the regular-grid errors are
low ulp counts against the multiprecision reference prices, but the correction
is aimed primarily at agreement with the J\"ackel double-precision price formula.
The polish is therefore best viewed as an optional reference-alignment step
rather than a universal improvement.
Production systems should preferably use the same high-accuracy J\"ackel-style
Black price for all pricing tasks, but many still use the textbook normal-CDF
Black--Scholes formula.  In those stacks, the extra polish can be below the
accuracy floor of the surrounding pricing code.

The main conclusion is not that a convergence proof makes implied-volatility
solving simple.  The proof guarantees the mathematical iteration once its
preconditions are met.  A robust production solver still needs many additional
steps: normalising quotes to equivalent out-of-the-money prices, evaluating
tail prices safely, guarding probabilities and updates, handling microscopic
near-at-the-money prices, treating prices close to the no-arbitrage upper bound,
and sometimes applying a final polish to machine-epsilon agreement with a chosen
reference price.
Writing such a solver is therefore less straightforward than ``start from a
guaranteed seed and iterate'', even when the real-arithmetic convergence is
settled.  The value of ThiopheneIV is that it separates these two issues: a
transparent monotone core, surrounded by the necessary engineering safeguards.
\appendixtitles{yes}
\FloatBarrier
\appendixstart
\appendix
\section[\appendixname~\thesection]{Benchmark Construction and Worked Examples}
\label{app:benchmark-construction}

\subsection{Dataset grids}

\begingroup
\sloppy
The dataset builders use the following grids, with the above filtering
applied after each raw case is generated:
\begin{itemize}
  \item \textbf{CLY-3D.} $S=100$, $r=0.03$,
    $K=\operatorname{linspace}(105,800,40)$,
    $T=\operatorname{linspace}(0.01,2,40)$, and
    $\sigma=\operatorname{linspace}(0.01,0.99,40)$, retaining prices
    above $10^{-20}$ as in the Cui--Liu--Yao comparison script.
  \item \textbf{CLY-20.} $S=100$, $r=0.03$, $\sigma=0.20$,
    $K=\operatorname{linspace}(105,180,40)$, and
    $T=\operatorname{linspace}(0.1,2,40)$.
  \item \textbf{CLY-80.} $S=100$, $r=0.03$, $\sigma=0.80$,
    $K=\operatorname{linspace}(105,800,40)$, and
    $T=\operatorname{linspace}(0.1,2,40)$.
  \item \textbf{J\"ackel.} $S=100$, $r=0$,
    $K=100\,\operatorname{linspace}(0.5,8,30)$,
    $T\in\{0.01,0.1,0.25,0.5,1,2\}$, and
    $\sigma=\operatorname{linspace}(0.02,4,30)$.
  \item \textbf{Market.} $S=100$, $r=0.03$,
    $K=100\,\operatorname{linspace}(0.7,1.5,30)$,
    $T\in\{1/252,5/252,21/252,63/252,0.5,1,2,5\}$, and
    $\sigma=\operatorname{linspace}(0.05,1.5,30)$.
  \item \textbf{Corners.} The union of low-volatility short-maturity
    ITM calls; deep OTM calls with $K\in\{200,300,500,1000,2000\}$;
    near-OTM short-maturity calls with
    $K=\operatorname{linspace}(101,150,10)$; high-volatility cases
    with $K\in\{100,150,200,500\}$; and the near-ATM small-price
    grid $K\in\{100.5,101,102,105,110\}$,
    $T\in\{0.001,0.005,0.01\}$,
    $\sigma\in\{0.005,0.01,0.02,0.05\}$.
  \item \textbf{Stress.} $S=100$, $r=0.03$,
    $K\in\{101,102,103,110,150,200,500,1000,2000,5000,10000\}$
    and $K\in\{10,20,50,80,90,95,98,99\}$,
    $T\in\{0.001,0.005,0.01,0.05,0.1,0.5,1,2,5,10\}$, and
    $\sigma\in\{0.01,0.02,0.05,0.10,0.20,0.30,0.50,0.80,0.99\}$.
  \item \textbf{HighVol.} $S=100$, $r=0$,
    $K\in\{1,2,3,4\}$,
    $T\in\{1,2,3,5,7,10\}$, and
    $\sigma\in\{0.5,0.8,1.0,1.2,1.5,2.0,2.5\}$, retaining only
    cases with finite $\ln\left(c\right)$ and $\ln\left(c\right)\le -0.05$ before the
    common normalisation filter.
\end{itemize}
\endgroup

\subsection{Normalisation and dataset construction}

All generated benchmark cases are converted to the same normalised OTM
representation before being passed to the solvers.  Given spot $S$,
strike $K$, maturity $T$, volatility $\sigma$, and rate $r$, the benchmark
helper computes
\[
  F = S e^{rT}, \qquad F_* = \min(F,K), \qquad
  K_* = \max(F,K), \qquad e^x = F_*/K_*,
\]
and then
\[
  x = \ln\left(e^x\right), \qquad v_{\rm ref}=\sigma\sqrt{T}.
\]
For accuracy reporting, the stored price is the OTM-normalised price obtained
from a multiprecision Black calculation and rounded to double precision.  For
the raw dataset filters, the
corresponding $\ln(c)$ is required to be finite, no larger than $-10^{-15}$
(effectively at the upper price boundary), no smaller than $-708$ (double
underflow), and $c\in(0,1)$.  The benchmark array stores
$(c,e^x,T,\sigma)$; all six generated datasets use calls, with ITM
calls represented through the OTM parity leg.

\subsection{Concrete worked examples}
\label{app:worked-examples}

For example, the near-ATM small-price Corners case
$S=100$, $K=100.5$, $T=0.01$, $\sigma=0.05$, $r=0$ gives
\[
\begin{aligned}
F &= 100, & e^x &= 100/100.5 = 0.9950248756218906,\\
x &= -0.004987541511039051, & v_{\rm ref} &= 0.005000000000000001,\\
\ln\left(c\right) &= -7.776203975967922, & c &= 4.196019744216237\times10^{-4}.
\end{aligned}
\]
It is therefore stored as a valid Corners case and enters the
near-ATM small-price seed branch ($c\le 5\times10^{-4}$ and
$|x|<0.01$), not the microscopic Bachelier branch
($c\le 10^{-6}$ and $|x|\le10^{-8}$).

A complementary saturated-price corner illustrates the conditioning at
high volatility close to ATM.  With
$S=100$, $K=100.01$, $T=10$, $\sigma=3$, and $r=0$, one obtains
\[
\begin{aligned}
e^x &= 0.9999000099990001,
& x &= -9.999500033332494\times10^{-5},\\
v_{\rm ref} &= 9.486832980505138,
& \ln\left(c\right) &= -2.1015432345450336\times10^{-6},\\
c &= 0.9999978984589737,
& 1-c &= 2.1015410263114376\times10^{-6}.
\end{aligned}
\]
This case is near ATM but enters the upper-tail branch because
$c>1/2$.  In this regime, a guarded production solver can still reprice to machine-precision price
accuracy even when the returned volatility differs from the generating value by
more than a few ulps.  The larger error in volatility is unavoidable conditioning:
the normalised vega is only $5.19\times10^{-6}$, so an absolute price
perturbation at double precision is amplified when inverted for~$v$.

A smaller but useful upper-half regression case is obtained from the exact
binary double inputs printed as
\[
  x=-10^{-6}, \qquad c=0.9999, \qquad T=1.
\]
The multiprecision reference total volatility for those rounded inputs is
\[
  v_*=7.7811840154613839563\ldots .
\]
ThiopheneIV returns
\[
  7.781184015461386,
\]
and ThiopheneIV+ returns the same value because the direct J\"ackel price polish
is inactive for $c>1/2$.  This is about $1.9$ ulps of the reference total
volatility.  If the expanded J\"ackel direct-price Newton correction is applied
nevertheless, it returns
\[
  7.781184015461925,
\]
about $609$ ulps from the same reference.  The direct price is accurate in
absolute terms, but above the midpoint the small quantity is the complementary
gap $1-c$; subtracting against the no-arbitrage upper bound makes a direct
price residual a poor Newton target.

\section{Guards and special branches}\label{sec:guards_and_branches}
\subsection{Microscopic Bachelier-limit branch}
\label{app:microscopic-bachelier}

The microscopic branch is deliberately narrower than the ordinary
near-ATM small-price seed.  It is entered only for normalised OTM prices
$c\le 10^{-6}$, $|x|\le 10^{-8}$, and $x\le0$ (with a one-part-in
$10^{12}$ tolerance on the price cut-off).  In this box the two Black
normal terms are almost equal and the volatility is of the same order as
the log-moneyness.  Iterating directly on the Black formula can therefore
spend its effort resolving a difference of nearly equal tail quantities,
while a central ATM approximation can be wrong by orders of magnitude as
soon as $|x|/v$ is no longer small.

The branch first removes the first-order lognormal skew by using the
sqrt-forward variables
\[
  \beta = c e^{x/2}, \qquad m=-x\ge0, \qquad a=m/v .
\]
The leading price is the Bachelier integral
\[
  I_0(m,v) = v\phi\left(a\right)-m\Phi\left(-a\right),
\]
and the local Black correction is evaluated through the expansion
\[
  \beta_{\mathrm{Bl}}(m,v)
  \approx I_0 - \frac{I_2}{8} + \frac{I_4}{128},
\]
where
\[
  I_2=\frac{v^3\phi\left(a\right)-m^2I_0}{3}, \qquad
  I_4=\frac{v^5\phi\left(a\right)-m^2I_2}{5}.
\]
The Newton denominator uses the exact sqrt-forward Black vega identity
\[
  \frac{\partial \beta}{\partial v}
  = \phi\left(a\right)e^{-v^2/8},
\]
so the correction remains well-scaled even when the price itself is
microscopic.

There are three exits.  If $m>0$ and $\ln\left(m/\beta\right)>20$, a normal
deep-tail equation is solved for $a=m/v$:
\[
  \ln\left(\frac{\beta}{m}\right)
  = -\frac{a^2}{2}-\frac{1}{2}\ln\left(2\pi\right)
    + \ln\left(\frac{1}{a}-\sqrt{\frac{\pi}{2}}\,
      \erfcx\left(\frac{a}{\sqrt{2}}\right)\right).
\]
The scaled Mills-ratio form avoids subtracting two underflow-level normal
tails.  This candidate is accepted only when the resulting ratio $m/v$ is
larger than $4$, where the omitted Black correction is numerically
negligible.  Otherwise, the rational Bachelier approximation of
Le~Floc'h~\cite{lefloc2016basispoint} supplies the normal-model seed,
and two Newton corrections are applied to the expansion above.  The only
remaining terminal path is defensive: if neither guarded path returns a
finite positive value, the branch returns the zero-volatility limit.  For
finite admissible inputs this path is not expected, since the rational seed
covers the ATM and transition regimes while the deep-tail solve intercepts
ratios below the rational approximation's finite domain.

A deep-tail example is obtained directly in normalised coordinates by taking
\[
  x=-10^{-14}, \qquad e^x=0.99999999999999, \qquad
  T=1, \qquad \ln\left(c\right)=-100.
\]
Equivalently, one may take $K_*=1$ and
$F_*=e^x=0.99999999999999$, giving
$C_{\mathrm{OTM}}=F_*c=3.7200759760207988\times10^{-44}$.
Thus $c=e^{-100}=3.720075976020836\times10^{-44}$ and
\[
  m=-x=10^{-14}.
\]
The Bachelier-limit branch is active because
$c\le10^{-6}$, $|x|\le10^{-8}$, and $x\le0$.  The deep-tail
sub-branch is selected first (before the rational Bachelier seed) because
\[
  \ln\left(m\right) - \ln\left(c\right) = 67.76380869808337 > 20.
\]
Solving the scaled Mills-ratio equation gives
\[
  \hat v = 9.155604419747184\times10^{-16}, \qquad
  m/\hat v = 10.922271803739784 > 4.
\]
A 120-digit check of the Black formula gives the consistent root
$v_{\rm ref}=9.155604419747139\times10^{-16}$ for
$\ln\left(c\right)=-100$.  The relative difference between the branch value and
this high-precision Black root is about $5.0\times10^{-15}$.  By
contrast, $v=1.25\times10^{-15}$ gives
$\ln\left(c\right)=-71.43799710528910$, so it is not the reference volatility for
the $\ln\left(c\right)=-100$ input.  In ordinary double precision the direct erfcx
Black log-price difference collapses for this point, which is why the
solver uses the Bachelier-limit tail equation instead of entering the
Black H3 path.

A second example shows why the rational Bachelier seed is still needed
outside the accepted deep tail.  Take
\[
  x=-10^{-8}, \qquad e^x=0.9999999900000001, \qquad
  T=1, \qquad c=10^{-16}.
\]
Then
\[
  m=10^{-8}, \qquad
  \beta=ce^{x/2}=9.999999950000000\times10^{-17}, \qquad
  \ln\left(m/\beta\right)=18.420680748952365.
\]
The price and moneyness are microscopic, but the conservative deep-tail
test is not satisfied because $\ln\left(m/\beta\right)<20$.  This is the transition
region: the central estimate $\sqrt{2\pi}\beta$ is too small by almost a
factor of eight, while the asymptotic tail equation is not used.  The
rational Bachelier seed gives the normal root
\[
  v = 1.9952018436169516\times10^{-9}, \qquad m/v=5.012024238044884,
\]
and the local Black correction leaves the value unchanged at the shown
precision.  A high-precision evaluation of
$\Phi\left(x/v+v/2\right)-e^{-x}\Phi\left(x/v-v/2\right)$ at this value gives
$1.0000000000000004\times10^{-16}$, i.e., the target price to double
precision.  The example is therefore neither an ATM case nor a safely
asymptotic tail case; it is exactly the narrow zone for which the
Bachelier rational seed keeps the expansion polish in its convergence
basin.

\subsection{Choi seed repair}
\label{app:floating-point-branches}

The Choi L3 formula is the intended non-Bachelier entry point for ThiopheneIV.
The seed-repair fallback in the implementation is used only if the Choi seed is
non-finite or non-positive after the floating-point probability calculation. It is not reached in the extensive benchmark presented in this paper. It is nevertheless part of ThiopheneIV as a production solver.  A
concrete adversarial example is $x\simeq -720$,
$e^x\simeq 2.03\times10^{-313}$, and $c=\operatorname{nextUp}(0)$: the L3
probability calculation becomes non-finite in double precision, and the repair
path supplies a finite positive seed before the cubic refinement begins.

The repair hierarchy is deliberately conservative.  For a very small near-ATM
price ($c<10^{-4}$ and $|x|<0.01$), it uses
\[
  v_0=\sqrt{x^2+2\pi c^2},
\]
the first-order ATM relation with a moneyness correction.  Otherwise it tries
the guarded OTM asymptotic seed
\[
  D=\sqrt{\max\{-2\ln(c)-\ln(2\pi),0\}},\qquad
  v_0=\frac{-2x}{D+\sqrt{D^2-2x}},
\]
accepted only when the discriminant and denominator are finite and positive.
If even this defensive path fails, the terminal fallback is
$v_0=\sqrt{2|x|}$, followed by the positive floor $10^{-10}$.  These branches
are outside the real-arithmetic convergence proof; their role is only to ensure
that rounded or adversarial inputs still reach a finite positive starting point.

\section[\appendixname~\thesection]{Expanded Normalised-Black Evaluation}
\label{app:expanded-black}

The cancellation-avoiding computation used for the expanded J\"ackel reference
price in Table~\ref{tab:accuracy} and
Figures~\ref{fig:pricing-formula-accuracy-diagnostic},
\ref{fig:pricing-formula-monotonicity-diagnostic}, and
\ref{fig:thiophene-polish-roundtrip-diagnostic}
is a price evaluator, not an additional implied-volatility iteration.  It works
in J\"ackel's sqrt-forward normalisation
\[
  \beta(x,s)=c(x,s)e^{x/2},
\]
and the OTM-forward-normalised price used elsewhere in this paper is recovered
afterwards as $c=\beta e^{-x/2}$.  For $x\le0$ and total volatility $s$, define
\[
  h=\frac{x}{s}, \qquad t=\frac{s}{2}, \qquad
  \nu(x,s)=\frac{1}{\sqrt{2\pi}}\exp\left[-\frac{1}{2}\left(h^2+t^2\right)\right].
\]
The exact beta-normalised Black price is
\[
  \beta(x,s)=e^{x/2}\Phi(h+t)-e^{-x/2}\Phi(h-t)
             =e^{ht}\Phi(h+t)-e^{-ht}\Phi(h-t),
\]
because $x=2ht$.  Its derivative with respect to total volatility is exactly
\[\partial\beta/\partial s=\nu(x,s).\]
J\"ackel's implementation therefore often evaluates the scaled price
$\beta/\nu$ in the difficult regions and multiplies by $\nu$ only at the end.
This is the source of the phrase ``expanded normalised Black'' in this paper.

The branch dispatcher is applied in the $(x,s)$ plane before any price is
formed.  With J\"ackel's constants $\eta=-13$ and
$\tau=2\epsilon^{1/16}$, where $\epsilon$ is double-precision machine epsilon,
Region~I is defined by
\[
  x<\eta s, \qquad
  s\left(\frac{s}{2}-\left(\tau+\frac{1}{2}+\eta\right)\right)+x<0,
\]
while Region~II is defined by
\[
  s(s-2\tau)-\frac{x}{\eta}<0 .
\]
The actual price evaluator is then
\[
  \beta_{\mathrm{J}}(x,s)=
  \begin{cases}
    \nu(x,s)\,R_{\mathrm{I}}(h,t), & \text{Region~I},\\[2pt]
    \nu(x,s)\,R_{\mathrm{II}}(h,t), & \text{Region~II},\\[2pt]
    \beta_{\mathrm{Cody}}(x,s), & \text{otherwise}.
  \end{cases}
\]
Region~I is tested first.  It is the deep-tail branch; it evaluates a rational
asymptotic expansion for the scaled price.  In the notation above it has the
form
\[
  R_{\mathrm{I}}(h,t)=\frac{t}{(h+t)(h-t)}\,\Omega(q,e),
  \qquad e=\left(\frac{t}{h}\right)^2,
  \qquad q=\left(\frac{h}{(h+t)(h-t)}\right)^2,
\]
where $\Omega$ is J\"ackel's Horner-evaluated polynomial in $q$, with
coefficients depending on $e$ and with the number of retained terms chosen by
the tail bound.  Region~II is the small-$t$ branch.  It is not a separate
model: it is the odd Taylor expansion of the same exact beta-normalised Black
price after the vega factor has been divided out,
\[
  R_{\mathrm{II}}(h,t)=t
  \sum_{j=0}^{6} b_j(h)t^{2j}.
\]
Let
\[
  a(h)=1+h\sqrt{\frac{\pi}{2}}\,
  \erfcx\left(-\frac{h}{\sqrt{2}}\right).
\]
Equivalently, $a(h)=1+h\Phi(h)/\phi(h)$, so the leading coefficient satisfies
$\partial\beta/\partial t\vert_{t=0}=2\nu(h,0)a(h)$.  This is why $a(h)$
appears in every coefficient.  For large negative $h$, the expression
$1+h\Phi(h)/\phi(h)$ is small and would be cancellation-prone if evaluated
literally; the implementation uses rational approximations to $a(h)$ in that
tail.
The coefficients are
\[
\begin{aligned}
b_0 &= 2a,\\
b_1 &= \frac{-1+a(3+h^2)}{3},\\
b_2 &= \frac{-7-h^2+a(15+10h^2+h^4)}{60},\\
b_3 &= \frac{-57-18h^2-h^4+a(105+105h^2+21h^4+h^6)}{2520},\\
b_4 &= \frac{1}{181440}\Bigl[-561-285h^2-33h^4-h^6\\
&\quad +a(945+1260h^2+378h^4+36h^6+h^8)\Bigr],\\
b_5 &= \frac{1}{19958400}\Bigl[-6555-4680h^2-840h^4-52h^6-h^8\\
&\quad +a(10395+17325h^2+6930h^4+990h^6+55h^8+h^{10})\Bigr],\\
b_6 &= \frac{1}{3113510400}\Bigl[-89055-82845h^2-20370h^4-1926h^6-75h^8-h^{10}\\
&\quad +a(135135+270270h^2+135135h^4+25740h^6+2145h^8+78h^{10}+h^{12})\Bigr].
\end{aligned}
\]

The fallback branch evaluates the same exact price with a Cody-style choice of
$\erfc$ or $\erfcx$.  Let
\[
  q_1=-\frac{h+t}{\sqrt{2}}, \qquad
  q_2=-\frac{h-t}{\sqrt{2}}, \qquad \rho=0.46875 .
\]
Then $2\beta_{\mathrm{Cody}}$ is computed as
\[
\begin{cases}
e^{x/2}\erfc(q_1)-e^{-x/2}\erfc(q_2),
  & q_1<\rho,\ q_2<\rho,\\[2pt]
e^{x/2}\erfc(q_1)-e^{-(h^2+t^2)/2}\erfcx(q_2),
  & q_1<\rho\le q_2,\\[2pt]
e^{-(h^2+t^2)/2}\erfcx(q_1)-e^{-x/2}\erfc(q_2),
  & q_2<\rho\le q_1,\\[2pt]
e^{-(h^2+t^2)/2}\left[\erfcx(q_1)-\erfcx(q_2)\right],
  & \rho\le q_1,\ \rho\le q_2 .
\end{cases}
\]
Each row is algebraically the same erfc representation of the Black price; the
branching only chooses where to factor out the Gaussian exponential so that the
two terms remain on a comparable numerical scale.

The optional ThiopheneIV+ polish uses this evaluator in the beta scale.  Given a
target $\beta_*=c_{\mathrm{target}}e^{x/2}$ and a current total volatility
$s$, one J\"ackel--Newton correction is
\[
  s_{\mathrm{new}}=s+\frac{\beta_* - \beta_{\mathrm{J}}(x,s)}{\nu(x,s)}.
\]
The correction is applied only on the lower-price half in the main solver; the
upper half is kept on the complementary $\ln(1-c)$ objective.

It is useful to separate these price-evaluation regions from the solver's
lower/upper objective split.
For fixed total volatility, the normalised OTM price increases with $x$ up to
the ATM boundary $x=0$.  Hence the largest possible price in Region~II is
approached at $x\to0^-$ and $s\to2\tau$, giving
\[
  \sup_{\mathrm{II}} c
  = 2\Phi(\tau)-1
  = \erf\left(\frac{\tau}{\sqrt{2}}\right)
  \approx 0.16650723223355586 .
\]
In Region~I, the price-maximising boundary gives
$d_1=x/s+s/2\le \eta+1/2+\tau$, and therefore
\[
  \sup_{\mathrm{I}} c
  = \Phi\left(\eta+\frac{1}{2}+\tau\right)
  = 5.1395297547074\times10^{-35} .
\]
Thus J\"ackel's explicit Region~I/II expansions are strictly lower-price
machinery; they never apply for $c\ge1/2$.  This is why the optional
ThiopheneIV+ direct-price polish is limited to the lower half and why the
upper half instead keeps the complementary $\ln(1-c)$ objective.

The expanded-J\"ackel price-evaluation regions used by the optional polish are
shown separately in Figure~\ref{fig:jaeckel-price-zones}.
\begin{figure}[!htbp]
\centering
\includegraphics[width=\textwidth]{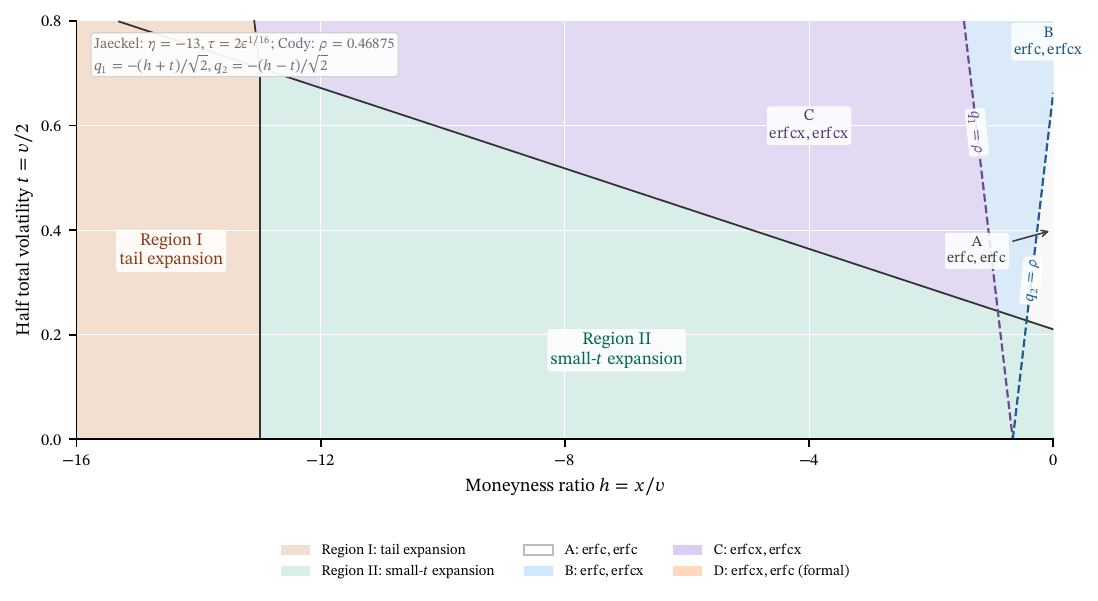}
\caption{Expanded-J\"ackel price-evaluation dispatch in the variables
$h=x/v$ and $t=v/2$.  Region~I uses the tail expansion, Region~II uses the
small-$t$ expansion, and the remaining domain uses the Cody erfc/erfcx path.
The ThiopheneIV+ polish uses this evaluator as its price target; the
ThiopheneIV core chooses its lower- or upper-tail objective independently of
these price-evaluation regions.
\label{fig:jaeckel-price-zones}}
\end{figure}

\section{Erfcx Implementation Impact}\label{sec:erfcx-impact}

\begin{table}[!htbp]
\caption{Accuracy against rounded multiprecision Black reference prices 
by dataset, using different implementations of \texttt{erfc} and \texttt{erfcx}.\label{tab:accuracy_erfcx}}
\scriptsize
\setlength{\tabcolsep}{4pt}
\begin{tabularx}{\textwidth}{lCCCCCCCC}
\toprule
& \textbf{CLY-3D} & \textbf{CLY-20} & \textbf{CLY-80}
& \textbf{J\"ackel} & \textbf{Market} & \textbf{Corners}
& \textbf{Stress} & \textbf{HighVol} \\
\midrule
\multicolumn{9}{l}{\textbf{ThiopheneIV accuracy -- max error (ulp of reference total volatility)}} \\[2pt]
Cody     &      190   &   52  &     8  &    142 &    520  &    557  &   285  &      2\\
Commons  &      133   &   49  &     8  &     63 &    177  &    329  &   138  &      2\\
Johnson  &      400   &   71  &    16  &     97 &   1210  &    537  &   484  &      3\\[4pt]
\multicolumn{9}{l}{\textbf{ThiopheneIV+ accuracy -- max error (ulp of reference total volatility)}} \\[2pt]
Cody       &     23   &    5  &     4   &    14  &    30  &     41   &   33   &     2\\
Commons    &     23   &    5  &     4   &    13  &    29  &     41   &   33   &     2\\
Johnson    &     22   &    5  &     7   &    13  &    29  &     41   &   32   &     3 \\
\bottomrule
\end{tabularx}
\end{table}
The polished ThiopheneIV+ path is robust across erfcx choices. The unpolished fast solver is where Apache Commons implementation is clearly better on our reference benchmarks (Table \ref{tab:accuracy_erfcx}).

\section[\appendixname~\thesection]{Choi--Huh--Su L3 Seed in ThiopheneIV}
\label{app:choi-l3}

ThiopheneIV uses the Choi--Huh--Su L3 lower-bound formula as an
initial value for the Euler--Chebyshev refinement.  This appendix records the formula
in the normalised variables used throughout.  The solver receives
$x=\ln\left(F_*/K_*\right)\le0$, $e^x=F_*/K_*$, and
$c=C_{\mathrm{OTM}}/F_*$.  Define
\begin{linenomath}
\begin{equation}\label{eq:choi-vars}
  k=-x\ge0, \qquad E=e^k=\frac{1}{e^x}=\frac{K_*}{F_*}.
\end{equation}
\end{linenomath}
For $k>0$, the L3 seed first maps the observed price to
\begin{linenomath}
\begin{equation}\label{eq:choi-probability}
  p_3 = \frac{c(c+E)}{2c+E-1}, \qquad z_3=\Phi^{-1}\left(p_3\right).
\end{equation}
\end{linenomath}
The denominator is positive on the admissible OTM price interval.  In
floating-point arithmetic the implementation still rejects non-positive
or non-finite denominators and clamps $p_3$ to the open interval
representable by doubles before evaluating the inverse normal CDF.

The final step is purely algebraic.  Since
\begin{linenomath}
\begin{equation}\label{eq:choi-d1}
  d_1(v)=\frac{x}{v}+\frac{v}{2}=-\frac{k}{v}+\frac{v}{2},
\end{equation}
\end{linenomath}
the seed solves $d_1(v)=z_3$, equivalently
$v^2/2-z_3v-k=0$.  The positive solution is evaluated as
\begin{linenomath}
\begin{equation}\label{eq:choi-root}
  v_{\mathrm{L3}} =
  \begin{cases}
    z_3+\sqrt{z_3^2+2k}, & z_3\ge0,\\[3pt]
    \dfrac{2k}{\sqrt{z_3^2+2k}-z_3}, & z_3<0,
  \end{cases}
\end{equation}
\end{linenomath}
where the second branch is the rationalised form of the same root and
avoids cancellation when $z_3$ is negative.

At the ATM limit $k=0$, Equation~\eqref{eq:choi-probability} reduces to
$p_3=(1+c)/2$ and the Black price is
$c=2\Phi\left(v/2\right)-1$.  Thus we use
\begin{linenomath}
\begin{equation}\label{eq:choi-atm}
  v_{\mathrm{ATM}}=2\Phi^{-1}\left(\frac{1+c}{2}\right),
\end{equation}
\end{linenomath}
with the small-price series
\begin{linenomath}
\begin{equation}\label{eq:choi-atm-series}
  v_{\mathrm{ATM}} \approx \sqrt{2\pi}\,c
  \left(1+\frac{\pi c^2}{12}\right), \qquad c<10^{-4},
\end{equation}
\end{linenomath}
to avoid spending an inverse-normal evaluation on an almost linear
case.  If any of these steps produces a non-finite or non-positive seed,
ThiopheneIV falls back to the same asymptotic safeguard used by the
standalone reference implementation.  For prices above the midpoint
($c>1/2$), the same seed is refined with the complementary
$\ln\left(1-c\right)$ objective.

\section[\appendixname~\thesection]{Monotone Convergence of Chebyshev's Method from the Choi Seed}
\label{app:chebyshev-convergence}

This appendix provides a self-contained proof that the Euler--Chebyshev method
converges monotonically from
below to the root of the log-price objective when started from the
Choi--Huh--Su L3 lower-bound seed, for all admissible OTM prices in
real arithmetic.  Throughout, we use the notation of
Section~\ref{sec:log-price-objective}.

\subsection*{Setup}

Let $x\le0$, $k=-x\ge0$, and let the log-price objective be
\begin{equation}
  g(v) = \ln b(x,v) - \ln c_{\mathrm{target}},
\end{equation}
where $b(x,v)$ is the normalised Black call price and
$c_{\mathrm{target}}\in(0,1)$ is fixed.  The unique positive root
$v_*$ satisfies $g(v_*)=0$.  The log-vega is $y(v)=g'(v)=f(v)/b(x,v)$,
where $f(v)=\phi(d_1(v))>0$ is the Black vega and $d_1(v)=x/v+v/2$.

For $v<v_*$ we have $g(v)<0$.  Define the Newton displacement
$\eta(v)=-g(v)/g'(v)>0$, and the normalised curvature
$q(v)=g''(v)/g'(v)$.  The \emph{Euler--Chebyshev step} is
\begin{equation}\label{eq:cheb-step}
  T(v) = v + \eta(v)\bigl(1 - \tfrac{1}{2}\,q(v)\,\eta(v)\bigr).
\end{equation}
Equivalently, with the standard Chebyshev--Halley curvature parameter
$\lambda(v)=g(v)g''(v)/[g'(v)]^2=-q(v)\eta(v)$, this is
$T(v)=v+\eta(v)(1+\lambda(v)/2)$.  The proof below keeps $q(v)$ because its
sign is used directly.

Let $v_0$ be the Choi--Huh--Su L3 seed
(Appendix~\ref{app:choi-l3}).  Choi et~al.~\cite{choi2023} prove the
lower-bound property: $0<v_0\le v_*$ for all admissible $(x,c)$.

\begin{Theorem}\label{thm:chebyshev-monotone}
For all $x\le0$, $v\in(0,v_*)$, and admissible
$c_{\mathrm{target}}\in(0,1)$, the Chebyshev displacement satisfies
\begin{equation}\label{eq:chebyshev-bound}
  0 < T(v) - v \le v_* - v.
\end{equation}
Consequently, the iterates $v_{n+1}=T(v_n)$ starting from $v_0$ are
monotone increasing, bounded above by $v_*$, and converge to $v_*$.
\end{Theorem}

The proof requires two lemmas.

\subsection*{Lemma~A: log-concavity of \texorpdfstring{$b$}{b} in \texorpdfstring{$v$}{v}}

\begin{Lemma}\label{lem:logconcave}
The normalised Black price $b(x,v)$ is log-concave in $v$:
$g''(v)\le0$ for all $x\le0$, $v>0$.
\end{Lemma}

\begin{proof}
Let $a(v)=f'(v)/f(v)$ denote the log-derivative of the vega.  Differentiating
$d_1=x/v+v/2$ gives $d_1'=-x/v^2+1/2$, hence
\begin{equation}
  a(v) = \frac{k^2}{v^3} - \frac{v}{4}.
\end{equation}
The second log-price derivative is $g''=g'(a-y)$.  Subtracting the two
log-derivative relations,
\begin{equation}
  q(v) = \frac{g''(v)}{g'(v)} = a(v) - y(v).
\end{equation}
Setting $H(v)=f(v)-a(v)b(x,v)$, we compute
\begin{equation}
  H'(v) = f'(v) - a'(v)b(x,v) - a(v)f(v) = \beta(v)\,b(x,v),
\end{equation}
where $\beta(v)=-a'(v)=3k^2/v^4+1/4>0$, using
$f'(v)=a(v)f(v)$ and $b'(x,v)=f(v)$.  As $v\downarrow0$,
$H(v)$ has a non-negative limit: for $k>0$, the standard Mills expansion
of the two cancelling normal tails gives $b(x,v)=f(v)(v^3/k^2+O(v^5))$,
so $H(v)\to0$; for $k=0$, $H(v)\to\phi(0)$.  Since $H'>0$, we conclude
$H(v)\ge0$, i.e.
$f(v)\ge a(v)b(x,v)$, equivalently $y(v)\ge a(v)$, equivalently
$q(v)=a(v)-y(v)\le0$.  Therefore $g''(v)=g'(v)q(v)\le0$.
\end{proof}

From Lemma~\ref{lem:logconcave}, $q(v)\le0$ and $\eta(v)>0$, so
$1-\frac12 q\eta\ge1>0$, hence $T(v)>v$.  This proves the left
inequality in~\eqref{eq:chebyshev-bound}.

\subsection*{Lemma~B: scalar curvature inequality}

\begin{Lemma}\label{lem:scalar}
For all $x\le0$, $v>0$, with $y(v)=g'(v)$, $a(v)=k^2/v^3-v/4$,
$\beta(v)=3k^2/v^4+1/4$,
\begin{equation}\label{eq:scalar}
  y^2 - 3ay + 2a^2 + \beta \ge 0.
\end{equation}
\end{Lemma}

\begin{proof}
Set $r=y-a\ge0$ (from Lemma~\ref{lem:logconcave}).  Substituting
$y=a+r$ into~\eqref{eq:scalar} gives
\begin{equation}
  y^2-3ay+2a^2+\beta = r^2-ar+\beta.
\end{equation}
If $a\le0$, then
\begin{equation}
  r^2-ar+\beta\ge\beta>0,
\end{equation}
so~\eqref{eq:scalar} follows immediately.  It remains to consider
$a>0$.  Define $A(v)=af(v)/(a^2+\beta)$.  Using $f'=af$ and
$a'=-\beta$, a direct calculation gives
\begin{equation}
  A'(v)-f(v)
  = -f(v)\,\frac{a\beta'+2\beta^2}{(a^2+\beta)^2}.
\end{equation}
Substituting $a=v(z-1/4)$, $\beta=3z+1/4$, $\beta'=-12z/v$ with
$z=k^2/v^4\ge0$ gives
\begin{equation}
  a\beta' + 2\beta^2
  = -12z^2 + 3z + 18z^2 + 3z + \tfrac{1}{8}
  = 6z^2 + 6z + \tfrac{1}{8} > 0.
\end{equation}
Hence $A'\le f$.  Because $a>0$ implies $k>0$, we also have
$a(u)>0$ for every $0<u\le v$ since $a'(u)=-\beta(u)<0$, and
$A(0+)=0$.  Integrating from $0$ to $v$ yields
$A(v)\le b(x,v)$, i.e.\ $af/(a^2+\beta)\le b$, i.e.\ $y=f/b\le (a^2+\beta)/a=a+\beta/a$,
hence $r\le\beta/a$.

\textbf{Case 1}: $\beta\ge a^2/4$.  Then $r^2-ar+\beta=(r-a/2)^2+\beta-a^2/4\ge0$.

\textbf{Case 2}: $\beta<a^2/4$.  On $r\in[0,\beta/a]$,
the parabola $D(r)=r^2-ar+\beta$ has its minimum at $r=a/2>\beta/a$, so $D$ is
decreasing on $[0,\beta/a]$.  Therefore
\begin{equation}
  D(r)\ge D\!\left(\tfrac{\beta}{a}\right)
  =\frac{\beta^2}{a^2}-\beta+\beta=\frac{\beta^2}{a^2}>0.
\end{equation}
\end{proof}

\subsection*{Proof of Theorem~\ref{thm:chebyshev-monotone}}

\begin{proof}
Let $\psi$ be the inverse branch of $g$ that maps $(-\infty,0)$ to
$(0,v_*)$.  Setting $u=g(v)<0$ and $s=-u>0$, we have $v=\psi(u)$ and
$v_*=\psi(0)$.  The inverse
derivatives are $\psi'=1/g'$, $\psi''=-g''/{g'}^3$, and
\begin{equation}
  \psi'''(u) = \frac{3{g''}^2 - g'g'''}{g'^5}.
\end{equation}

Observe that the Chebyshev displacement equals the second-order Taylor
approximation of $\psi$ from $u$ to $0$:
\begin{align}
  T(v)-v
  &= s\psi'(u) + \tfrac{1}{2}s^2\psi''(u) \notag\\
  &= \frac{-g}{g'} - \frac{g''g^2}{2g'^3}
   = \eta\bigl(1-\tfrac{1}{2}q\eta\bigr).
\end{align}
Taylor's theorem gives
\begin{equation}
  v_*-v = \psi(0)-\psi(u)
  = s\psi'(u)+\tfrac{1}{2}s^2\psi''(u)+\tfrac{1}{6}s^3\psi'''(\xi)
  = T(v)-v + \tfrac{s^3}{6}\psi'''(\xi),
\end{equation}
for some $\xi\in(u,0)$.  Thus $T(v)-v\le v_*-v$ follows once
$\psi'''\ge0$ on $(u,0)$.  Since the formula for $\psi'''$ is evaluated
at $w=\psi(\xi)\in(v,v_*)$, it is enough to verify
$3g''^2-g'g'''\ge0$ for all positive total volatilities.  Using
$\delta_2=g''/g'=q$ and
$\delta_3=g'''/g'$ from Proposition~\ref{prop:derivatives},
\begin{equation}
  3g''^2 - g'g''' = g'^2(3q^2-\delta_3).
\end{equation}
From~\eqref{eq:d2}--\eqref{eq:d3},
\begin{equation}
  \delta_3 = \frac{-3h^2-t^2+(h^2-t^2)^2}{v^2} - 3yq - y^2.
\end{equation}
Substituting $a = (h^2-t^2)/v = k^2/v^3 - v/4$, $a'=-\beta$,
$q=a-y$, and $\delta_3 = a' + a^2 - 3yq - y^2 = -\beta+a^2-3yq-y^2$
yields, after collecting terms,
\begin{equation}
  3q^2-\delta_3 = 3(a-y)^2-(-\beta+a^2-3y(a-y)-y^2)
  = 2q^2 - q' = y^2-3ay+2a^2+\beta.
\end{equation}
By Lemma~\ref{lem:scalar}, this is non-negative for all $x\le0$, $v>0$.
Hence $\psi'''(\xi)\ge0$, and the right inequality
$T(v)-v\le v_*-v$ holds.

The left inequality $T(v)>v$ follows from $q\le0$ (Lemma~\ref{lem:logconcave})
and $\eta>0$.

Both inequalities proven, the iterates satisfy
$v_0\le v_1\le v_2\le\cdots\le v_*$.  The sequence is monotone and bounded
above, so it converges to some $\ell\le v_*$.  If $\ell<v_*$ then $g(\ell)<0$
and $T(\ell)>\ell$, contradicting $v_n\to\ell$ and $v_{n+1}=T(v_n)\to\ell$.
Hence $\ell=v_*$.

Local cubic convergence follows from the standard Chebyshev error
expansion, since $g'(v_*)>0$ guarantees a simple root.
\end{proof}

\begin{Remark}
The proof is entirely real-analytic and holds for the lower-tail log-price
objective for all
$x\le0$, $v>0$, $c\in(0,1)$, with no compact-domain restriction.  The
only prerequisite beyond Choi's lower-bound seed
(Appendix~\ref{app:choi-l3}) is the log-concavity of $b$ in $v$
(Lemma~\ref{lem:logconcave}) and the scalar curvature
inequality~\eqref{eq:scalar} (Lemma~\ref{lem:scalar}).  Neither lemma
requires any upper-price exclusion; the saturated-price regime
$c\approx1$ is covered because both lemmas hold pointwise for all
admissible $(x,v)$ for the $\ln(c)$ transform.  The complementary
$\ln(1-c)$ branch is a separate conditioning choice and is not a consequence of
this monotonicity theorem.
\end{Remark}

\section[\appendixname~\thesection]{Monotone Halley Convergence for the Complementary Objective}
\label{app:loggap-halley-convergence}

This appendix records the complementary analogue that is available when the
upper-tail objective is refined by Halley's method rather than by the
Euler--Chebyshev polynomial correction.  It is stated separately because its
proof uses the negative Schwarzian of the complementary log objective, not the
lower-tail log-concavity and inverse-curvature argument of
Appendix~\ref{app:chebyshev-convergence}.
The Schwarzian enters only to prove that Halley is Newton's method applied to a
concave transformed residual; convexity of the original complementary log
objective fixes the sign of the step but is not enough by itself to rule out
overshoot.

Let $b(x,v)$ denote the normalised OTM Black price and set
\begin{equation}
  Q(x,v)=1-b(x,v).
\end{equation}
For a target $Q_*\in(0,1)$ define the increasing complementary log objective
\begin{equation}\label{eq:halley-loggap-objective}
  G(v)=\ln Q_* - \ln Q(x,v).
\end{equation}
The root $v_*$ is characterised by $Q(x,v_*)=Q_*$.  For $v<v_*$ we have
$Q(x,v)>Q_*$ and therefore $G(v)<0$.

Define
\begin{equation}
  \eta(v)=-\frac{G(v)}{G'(v)},\qquad R(v)=\frac{G''(v)}{G'(v)}.
\end{equation}
The standard Chebyshev--Halley curvature parameter is
$\lambda(v)=G(v)G''(v)/[G'(v)]^2=-R(v)\eta(v)$.
The Halley map can then be written in the equivalent forms
\begin{equation}\label{eq:halley-loggap-map}
  T_H(v)
  =v-\frac{G(v)/G'(v)}{1-\tfrac12G(v)G''(v)/[G'(v)]^2}
  =v+\frac{\eta(v)}{1-\tfrac12\lambda(v)}
  =v+\frac{\eta(v)}{1+\tfrac12 R(v)\eta(v)}.
\end{equation}

\begin{Theorem}\label{thm:loggap-halley-monotone}
For every $x\le0$, every $Q_*\in(0,1)$, and every $v\in(0,v_*)$, the Halley
map~\eqref{eq:halley-loggap-map} satisfies
\begin{equation}
  0<T_H(v)-v\le v_*-v.
\end{equation}
Consequently, the iterates $v_{n+1}=T_H(v_n)$ from any starting value
$v_0\in(0,v_*)$ are monotone increasing, bounded above by $v_*$, and converge
to $v_*$.  The local convergence at the root is cubic.
\end{Theorem}

\begin{proof}
Write $k=-x\ge0$, $d_1=x/v+v/2$, and
\begin{equation}
  f(v)=\phi(d_1(v)),\qquad y(v)=\frac{f(v)}{Q(x,v)}.
\end{equation}
Since the Black vega is $\partial b/\partial v=f(v)$, we have
$Q'(x,v)=-f(v)$ and therefore
\begin{equation}
  G'(v)=y(v)>0.
\end{equation}
In particular, $Q(x,v)$ is continuous and strictly decreasing from $1$ as
$v\downarrow0$ to $0$ as $v\to\infty$, so the equation $Q(x,v)=Q_*$ has a
unique solution $v_*$ for every $Q_*\in(0,1)$.
Let
\begin{equation}
  a(v)=\frac{f'(v)}{f(v)}=\frac{k^2}{v^3}-\frac{v}{4},
  \qquad
  \beta(v)=-a'(v)=\frac{3k^2}{v^4}+\frac14>0.
\end{equation}
Differentiating $y=f/Q$ gives
\begin{equation}\label{eq:loggap-convex-ratio}
  \frac{G''(v)}{G'(v)}=a(v)+y(v).
\end{equation}

It remains to show that the right-hand side of
\eqref{eq:loggap-convex-ratio} is non-negative.  Define
\begin{equation}
  H(v)=f(v)+a(v)Q(x,v).
\end{equation}
Using $f'=af$, $Q'=-f$, and $a'=-\beta$, we obtain
\begin{equation}
  H'(v)=f'(v)+a'(v)Q(x,v)+a(v)Q'(x,v)=-\beta(v)Q(x,v)<0.
\end{equation}
The Mills expansion of the two complementary normal tails gives
$Q(x,v)=f(v)(1/d_1+1/(-d_2)+o(1/v))$ as $v\to\infty$, where
$d_2=d_1-v$.  Since $a(v)\sim -v/4$, this implies $H(v)\to0$ as
$v\to\infty$.  Because $H$ is strictly decreasing and tends to zero at
infinity, $H(v)>0$ for every finite $v>0$.  Hence
\begin{equation}
  R(v)=\frac{G''(v)}{G'(v)}=a(v)+y(v)=\frac{H(v)}{Q(x,v)}>0.
\end{equation}
Thus $G$ is increasing and convex.

We next show that the Schwarzian derivative of $G$ is negative.  Let
\begin{equation}
  W(v)=\sqrt{a(v)^2+2\beta(v)},\qquad M(v)=\frac{f(v)}{W(v)}.
\end{equation}
Since $Q'(x,v)=-f(v)$ and $Q(x,v)\to0$ as $v\to\infty$, we have
\begin{equation}
  Q(x,v)=\int_v^\infty f(u)\,du.
\end{equation}
Moreover $M(v)\to0$ as $v\to\infty$.  A direct differentiation gives
\begin{equation}\label{eq:hazard-bound-derivative}
  M'(v)+f(v)=\frac{f(v)}{W(v)^3}
  \left[W(v)^2\bigl(W(v)+a(v)\bigr)+a(v)\beta(v)-\beta'(v)\right].
\end{equation}
The bracket is positive.  If $a\ge0$, then all terms in the bracket are
non-negative and $-\beta'=12k^2/v^5\ge0$.  If $a<0$, write $s=-a>0$.
Since $W^2=s^2+2\beta$,
\begin{equation}
  W^2(W-s)-s\beta
  =\beta\left(\frac{2W^2}{W+s}-s\right)
  =\beta W\frac{2-s/W-(s/W)^2}{1+s/W}>0,
\end{equation}
and again $-\beta'\ge0$.  Hence $M'(v)+f(v)>0$.  Integrating this inequality
from $v$ to infinity yields
\begin{equation}
  \frac{f(v)}{W(v)}=M(v)\le Q(x,v),
  \qquad\text{so}\qquad
  y(v)=\frac{f(v)}{Q(x,v)}\le W(v).
\end{equation}
Therefore
\begin{equation}\label{eq:loggap-schwarzian-negative}
  \mathcal{S}G(v)
  =\frac{G'''(v)}{G'(v)}-\frac{3}{2}\left(\frac{G''(v)}{G'(v)}\right)^2
  =\frac12\left(y(v)^2-a(v)^2\right)-\beta(v)\le0.
\end{equation}
The inequality is strict for finite $v$ because $M'(v)+f(v)>0$ on every
non-empty interval.

Finally set
\begin{equation}
  \Psi(v)=\frac{G(v)}{\sqrt{G'(v)}}.
\end{equation}
Halley's method applied to $G$ is Newton's method applied to $\Psi$:
\begin{equation}
  v-\frac{\Psi(v)}{\Psi'(v)}
  =v-\frac{G(v)/G'(v)}{1-\tfrac12G(v)G''(v)/[G'(v)]^2}
  =T_H(v).
\end{equation}
For $v<v_*$, $G(v)<0$ and $R(v)>0$, so
\begin{equation}
  \Psi'(v)=\sqrt{G'(v)}
  \left(1-\frac12\frac{G(v)G''(v)}{[G'(v)]^2}\right)>0.
\end{equation}
Furthermore,
\begin{equation}
  \Psi''(v)=-\frac12\frac{G(v)}{\sqrt{G'(v)}}\,\mathcal{S}G(v)<0,
\end{equation}
by~\eqref{eq:loggap-schwarzian-negative}.  Thus $\Psi$ is increasing and
concave on $(0,v_*)$, with $\Psi(v)<0$ and $\Psi(v_*)=0$.  Concavity gives
\begin{equation}
  0=\Psi(v_*)\le \Psi(v)+\Psi'(v)(v_*-v),
\end{equation}
so the Newton displacement for $\Psi$ satisfies
\begin{equation}
  0<T_H(v)-v=-\frac{\Psi(v)}{\Psi'(v)}\le v_*-v.
\end{equation}

The iterates are therefore monotone increasing and bounded above by $v_*$, so
they converge to some limit $\ell\le v_*$.  If $\ell<v_*$, continuity of the
Halley map and the strict positivity of the displacement on $(0,v_*)$ would
give $T_H(\ell)>\ell$, contradicting $v_{n+1}-v_n\to0$.  Hence $\ell=v_*$.
The usual Halley error expansion applies at the simple root, where
$G'(v_*)>0$, and gives cubic local convergence.
\end{proof}

\vspace{6pt}

\authorcontributions{Not applicable (single author).}

\funding{This research received no external funding.}

\dataavailability{The source code used to generate the benchmark
tables, including the solver variants, benchmark harnesses, and dataset
builders, is available in the accompanying source distribution.  The exact
dataset grids are reproduced in Appendix~\ref{app:benchmark-construction}.}

\acknowledgments{The author thanks Gary Kennedy for a thorough review and
excellent feedback.}

\conflictsofinterest{The author declares no conflicts of interest.}

\begin{adjustwidth}{-\extralength}{0cm}

\reftitle{References}

\end{adjustwidth}

\end{document}